%% file: arxiv_temporal_clustering.tex
\newif\ifdraft
\draftfalse

\newif\iftechreport
\techreporttrue
\documentclass[journal]{IEEEtran}
\usepackage{setspace}
\usepackage{cite}
\usepackage{color}
\usepackage{graphicx}
\usepackage[cmex10]{amsmath}
\usepackage{amsthm,amsfonts}
\usepackage{algorithmic}
\usepackage{comment}

\newtheorem{theorem}{Theorem}
\newtheorem{definition}{Definition}
\newtheorem{remark*}{Remark}
\newtheorem{lemma}{Lemma}

\newcommand{\mycomment}[3]%
{\ifdraft \textcolor{#2}{{\bf\textsc{#1}:}~~#3} \else \fi}

\newcommand{\vikas}[1]%
{\mycomment{vikas}{blue}{#1}}

\newcommand{\yudong}[1]% 
{\mycomment{yudong}{green}{#1}}

\newcommand{\rahul}[1]% 
{\mycomment{rahul}{magenta}{#1}}

\begin{document}

\title{Detecting Overlapping Temporal Community Structure in Time-Evolving Networks}

\author{
  \IEEEauthorblockN{Yudong Chen\IEEEauthorrefmark{1}, \IEEEauthorblockN{Vikas
  Kawadia}\IEEEauthorrefmark{2} and \IEEEauthorblockN{Rahul
Urgaonkar}\IEEEauthorrefmark{2} }\\
\IEEEauthorblockA{\IEEEauthorrefmark{1}The University of Texas at Austin, ydchen@utexas.edu}\\
\IEEEauthorblockA{\IEEEauthorrefmark{2}Raytheon BBN Technologies, \{vkawadia,rahul\}@bbn.com}
}

\maketitle

\input{abstract}

%\IEEEpeerreviewmaketitle

\input{intro}

\input{related}

\input{model}

\input{expt}
\input{applications}

\input{conclusion}

\subsection*{Acknowledgements}

\small{Research was sponsored  by the Army Research Laboratory and was accomplished
under Cooperative Agreement Number W911NF-09-2-0053.}
%The views and conclusions
%contained in this document are those of the authors and not of the sponsors.}
%should not be
%interpreted as representing the official policies, either expressed or implied,
%of the Army Research Laboratory or the U.S. Government. The U.S. Government is
%authorized to reproduce and distribute reprints for Government purposes
%notwithstanding any copyright notation here on.

\bibliographystyle{./IEEEtranBST1/IEEEtran}
%\vspace{-0.08in}
%\bibliographystyle{abbrv}
\bibliography{cydong1,communities,applications}

\appendix

\input{fast_algorithm}

\iftechreport
\input{proof}

\fi

\end{document}

%% file: abstract.tex
\begin{abstract}
We present a principled approach for detecting overlapping temporal
community structure in dynamic networks. Our method is based on the following
framework: find the overlapping temporal community structure that maximizes a
quality function associated with each snapshot of the network subject to a
temporal smoothness constraint. A novel quality function and a smoothness
constraint are proposed to handle overlaps, and a new convex relaxation is used
to solve the resulting combinatorial optimization problem.  We provide
theoretical guarantees as well as experimental results that reveal community
structure in real and synthetic networks. Our main insight is that certain
structures can be identified only when temporal correlation is considered
\emph{and} when communities are allowed to overlap.  In general, discovering
such overlapping temporal community structure can enhance our understanding of
real-world complex networks by revealing the underlying stability behind their
seemingly chaotic evolution. 
%This has applications to
%communication networks  such as content caching, routing in
%mobile wireless networks, and developing realistic mobility models.
\end{abstract}

%% file: intro.tex
\section{Introduction}

Communities are densely connected groups of nodes in a network. Community
detection, which attempts to identify such communities, is a fundamental
primitive in the analysis of complex networked systems that span multiple
disciplines in network science such as biological networks, online social networks, epidemic
networks, communication networks, etc. It serves as an important tool for
understanding the underlying, often latent, structure of networks and has a wide
range of applications: user profiling for online marketing, computer virus
spread and spam detection, understanding protein-protein interactions, content
caching, to name a few.  The concept of communities has been generalized to
overlapping communities which allows nodes to belong to multiple communities at
the same time. This has been shown to reveal the latent structure at multiple levels
of hierarchy.

Community detection in static networks has been studied extensively (see ~\cite{Fortunato2010community} for a
comprehensive survey), but has primarily been  applied to social networks and
information networks.  Applications to communications networks have been
few.
%\footnote{In Section~\ref{sec:apps}, we discuss in detail three such
%applications: Routing in disruption tolerant wireless networks, content caching
%in networks, and developing realistic mobility models.}.  
Perhaps this is
because communication networks (and in particular wireless networks) change at a
much faster timescale than social and information networks, and the science of
community detection in time-varying networks is still developing.  In this paper, 
we hope to narrow this gap by providing efficient techniques for detecting communities in
networks that vary over time while allowing such communities to be
overlapping as we elaborate below.

Temporal community detection \cite{evolutionaryclustering, Facetnet ,
kim2009particle, estrangement} aims to identify how communities emerge, grow,
combine, and decay \emph{over time}. This is useful in practice because most
networks of interest, particularly communication networks, are time-varying.
Typically, temporal community detection enforces some continuity or
``smoothness'' with the past community structure as the network evolves. While
one could apply static community detection independently in each snapshot, this
fails to discover small, yet persistent communities because, without the
smoothness constraints, these structures would  be
buried in noise and thus be unidentifiable to static methods.

In this paper, we propose a principled framework that goes beyond 
regular temporal communities and incorporates the concept of \emph{overlapping temporal
communities (OTC)}. In this formulation, nodes can belong to multiple communities at
any given time and those communities can persist over time as well.   This
allows us to detect even more subtle persistent structure that would otherwise
be subsumed into larger communities. We illustrate the various notions of
community structure in Fig.~\ref{fig:schematic}. As we will demonstrate in
Section~\ref{sec:expt} using both synthetic and real-world data, our framework is
able to correctly identify such community structure.

\begin{figure}[!t]
\centering
\includegraphics[width=0.8\linewidth, height=0.6\linewidth]{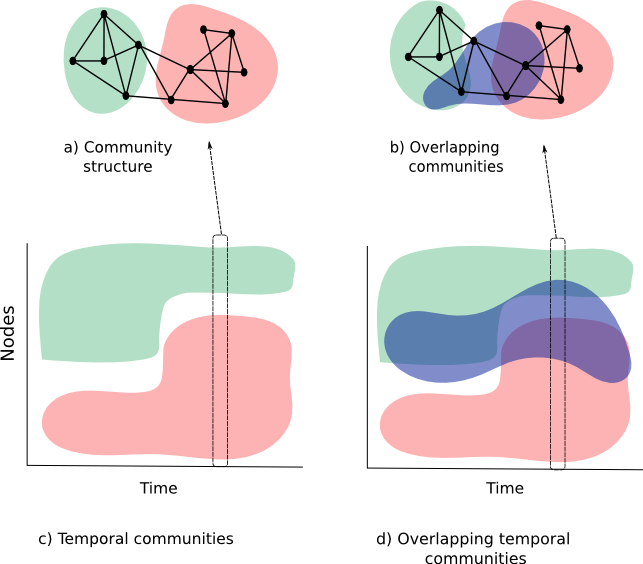}
\caption{A schematic illustrating the various notions of community structure in
networks. Panel a) shows a typical community structure in an example network. If
one uses a quality function and methods that look for overlapping community
structure on the same network, one could find a structure shown in panel b).
When the network is time-varying, we illustrate the temporal community structure
by showing the communities that a node belong to over time. The top panels
are single snapshots of the network evolution in the bottom panels. The
(non-overlapping) temporal community structure in panel c) reveals how
communities change with time. The overlapping temporal community structure shown
in panel d) can uncover deeper hidden patterns such as small communities
persistent over time (shown in blue). }

\label{fig:schematic}
\end{figure}

Knowing the OTC structure of an evolving network 
is useful because, although the observed network may change rapidly, its latent
structure often evolves much more slowly. For example, contact-based
social networks might change from day to day due to people's varying
daily activities, but the social groups (e.g. family, friends, colleagues)
that people belong to are much more stable. Identifying such latent time-persistent
structure can reveal the fundamental rules governing the seemingly chaotic
evolution of real-life complex networks. In addition, knowledge of the
times when significant changes occur could be used for predicting
network evolution.

The OTC structure of networks has many applications to designing communication
networks as well. For example, it can be used in efficient distributed storage of
information in a wireless network so that average access latency is
minimized~\cite{slinky}. The OTC structure can also be used to select
relay nodes and design routing schemes in a disruption tolerant wireless
networks. Another application is to devise real-life mobility models for
analysis and evaluation of network protocols.  We elaborate on these
applications in Section~\ref{sec:apps}.

\textbf{Our approach:} We describe the key ideas behind our techniques.  A naive
approach to temporal community detection is to perform static community
detection independently in each snapshot. The limitations of this approach are
well-documented~\cite{estrangement}, as it is very sensitive to even minor
changes in the network. The approach can be extended to detect OTC structure as
well but has the same limitations. Past work, including
\cite{evolutionaryclustering} and  \cite{estrangement} argue that temporal
communities can be detected if an explicit smoothness constraint that captures
the distance to past partitions is enforced. With the smoothness constraint, it
is possible to go beyond static methods and detect small persistent communities,
as information at multiple snapshots is considered \emph{together}.  In this
paper, we build on the above intuition and propose an approach for detecting OTC
structure using temporal information. Our approach is a novel convex relaxation
of the following combinatorial problem: find the temporal community structure
that maximizes a quality function associated with each snapshot subject to a
temporal smoothness constraint.  To handle overlapping communities, we use
generalizations of the quality function proposed in ~\cite{chen2012overlap} and
the temporal smoothness measure proposed in ~\cite{estrangement}.  While the
quality function favors densely connected groups, the smoothness constraint
promotes persistent structure. Our formulation is fairly general and allows
other quality functions and smoothness metrics to be used.  Further, it is
naturally applicable to overlapping temporal communities and does not require
any ad hoc modifications.  Unlike most existing approaches that use greedy
heuristics to solve the resulting NP-hard problem of optimizing over the
combinatorial set of all partitions (or covers when overlaps are allowed), we
use a tight convex relaxation of this set via the trace norm. This not only
results in a convex optimization problem that can be solved efficiently using
existing techniques, but also enables us to obtain \emph{a priori} guarantees on
the performance of our method, and provides valuable insight. In particular, our
analysis shows that, under a natural generative model, our method is able to
recover communities that persists for $m$ snapshots and have size
$K\ge\sqrt{\frac{n}{m}}$, where $n$ is the number of nodes. This result
highlights the \emph{benefit} of utilizing temporal information: with more
snapshots, we are able to detect smaller communities. We believe this specific
relation is novel, and applies beyond the particular methods in this
paper. We note that our approach can detect non-overlapping temporal
communities as well.

%as well as overlapping
%temporal communities depending on the quality function and the smoothness measure
%used. To handle overlapping temporal communities, we use the quality function
%proposed in ~\cite{chen2012overlap} and a generalization of the temporal
%smoothness measure proposed in ~\cite{estrangement}. Thus our method is
%naturally applicable to overlapping temporal communities by  generalizing the
%quality function and the smoothness measure and does not require any ad hoc
%modifications.

%Many existing works
%on temporal community detection thus focus on matching the performance
%of static community detection in both speed and quality, and improve
%stability of the results. 

To summarize, we provide the first principled formulation
of the problem of detecting overlapping temporal communities. A critical piece
of the formulation is the quality function for quantifying the community structure
in any snapshot, and a distance function to ensure contiguity with the past
community structure. To the best of our knowledge, we are the first ones to
propose such functions for overlapping communities. We provide a
convex relaxation and hence an efficient way to solve the optimization problem,
while most existing work relies on greedy heuristics.  In addition, we provide theoretical
guarantees on the performance of our method, and we discuss the insights we gain from the guarantees. Finally, we evaluate our method using
several synthetic and real network data-sets and illustrate its efficacy. We also discuss
applications of our method to communication networks. 

%We hope that efficient
%techniques that we provide for detecting structure in time-varying networks
%enable many more applications to communication networks.

%The rest of the paper is organized as follows. We review related work in the next section and
%present 

%In the subsequent sections, we briefly review related work in the
%next section, and present our formulation.

\textbf{Remark on terminology: } In the sequel, we use \emph{cluster} and \emph{community} interchangeably, both allowing overlaps.

%% file: related.tex
\section{Related Work}

There is a long line of work on community detection which has been
comprehensively surveyed in~\cite{Fortunato2010community}. 
Here, we focus on work that is most relevant to our approach.
In particular, \cite{agarwal2008modularity} presents a convex formulation for
optimizing modularity~\cite{Newman2006}, a well-known quality function for static
non-overlapping communities. Our convex formulation is completely different,
and our allow overlapping and temporal communities. 
 %\vikas{Not sure how to cite \cite{chen2012overlap}}
Our formulation is related to low-rank matrix recovery techniques~\cite{candes2009robustPCA,chandrasekaran2011siam}. This line of work typically
uses trace norm as a convex surrogate of the non-convex rank function, and
similarly the $ \ell_1 $ norm for sparsity. In this paper, we use trace
norm as a relaxation for the set of covers, a combinatorial and
non-convex set, and the (weighted) $ \ell_1 $ norm as the quality
function. Similar relaxation for static clustering without overlaps is considered
in~\cite{chen2012sparseclustering}.
Our approach for dealing with overlaps is very different from exiting work; for a survey
we refer to \cite{overlapping2012Survey} and
section XI of \cite{Fortunato2010community}. 
In the rest of this section we focus on an overview of temporal clustering.

Most existing work on temporal clustering can be divided into two categories: 1) maximize quality
function subject to smoothness constraints; 2) slightly modify the clustering structure
from the previous snapshot.

The first approach starts with \cite{evolutionaryclustering}, which proposes the framework of
Evolutionary Clustering that aims to optimize a combination of a snapshot
quality and a temporal smoothness cost. The work in \cite{estrangement} argues
that a specific choice of the temporal cost, namely estrangement, works well.
In \cite{Facetnet} the authors uses the KL-divergence in both the snapshot quality and
temporal cost, and reformulates the problem using non-negative matrix
factorization in order to obtain soft clusterings. In \cite{kim2009particle}
a particle-density based method is proposed to optimize the clustering objective. 
All of these works use greedy algorithms to solve the optimization problem, which only guarantees convergence to a local optimum. Our formulation is similar to the Evolutionary Clustering framework, but we are able to use 
convex optimization via a reasonable relaxation.

The second class of methods typically work as follows: each time the network changes, they
modify the clustering structure to reflect the change according to some predefined
rules. Smoothness is maintained since the modification would not change the
clustering too much. For example, \cite{cazabet2010overlap} proposes iLCD (intrinsic
Longitudinal Community Detection) algorithm, which updates, merges, and creates
communities based on the previous clustering; overlap is allowed.
\cite{nguyen2011adaptive} adopts a similar idea, but does not allow a node to
belong to multiple clusters while follow-up work \cite{nguyen2011overlapping} removes this restriction. However, all of these works use update rules that are based on heuristics; some of them might produce duplicates or very small communities and need to use ad hoc procedures to remove them. Unlike our work, they do not provide any analytical guarantees.

%Other existing approaches include \cite{Graphscope} and \cite{timefall} 
%which propose a
%parameter-free algorithm to find the non-overlapping clustering structure that best compresses
%the data based on the Minimum Description Length (MDL) principle, as well as
%\cite{Mucha2010}, which used a multi-slice modularity function defined on a
%stacked network that couples all snapshots. 

Other existing approaches include \cite{Graphscope,timefall,Mucha2010},
which use objective functions that essentially measure both snapshot quality and temporal smoothness.
Also, \cite{chi2009evolutionary}
propose a method to detect communities in multi-dimensional networks. None of
these however detect overlapping communities and simultaneously detect their
evolution over time.

%\todo{Other:    \cite{palla2005overlap} \cite{palla2007quantifying}  \cite{tantipathananandh2009constant}}

%% file: model.tex
\section{Formulation and Algorithm}

\label{sec:model}

In this section, we formulate the problem and describe our algorithm. We consider the following natural formulation of OTC detection. Suppose we are given $T$
snapshots of a network with $ n $ nodes in terms of adjacency matrices $A^{t}$,
$t=1,\ldots,T$ \footnote{We use the convention $A_{ii}^{t}=1$.}. Our general
goal is, at each time $ t $, to assign each node to a number of clusters so that
the edge densities within clusters are higher than those across clusters, and
that the assignment does not change rapidly with time. Note that each node might
be assigned to multiple clusters, and clusters can overlap. A node might also be 
associated with no cluster; these nodes are called outliers, and are common in real networks.
Mathematically,  let $ r^t $ be the
total number of communities at time $ t $. the value of $r^t$ is, of course, not known a priori.
We would like to find $ T $ \emph{covers with outliers}, where a
cover $\mathcal{C}^t$ with outliers  means a collection of $ r^t $ subsets $ \mathcal{C}^t =
\{C^t_k, k=1,\ldots,r^t \}$ with $C^t_k\subseteq \{1,\ldots,n\} $; again note that we allow outlier nodes that do not belong to any of the subsets. For convenience, we will use cover in the sequel when we actually mean cover with outliers.

To make this formulation concrete, there are several questions that need to be
answered. 1) How to concisely represent a cover? 2) When overlaps are allowed, how
to measure the quality of a cover? 3) In particular, how to avoid degenerate
solutions? For example, declaring each edge as a cluster would make
the in-cluster edge density $ 1 $ and across-cluster density $ 0 $, but is
an undesirable solution providing little information. Similarly, 
producing a cluster that differs from another only by one node hardly reveals
any additional structure. 4) How to enforce temporal smoothness when overlap
is present? 5) How to solve the resulting optimization problem over covers?

In the remainder of this section, we present our precise approach and address
the above questions.

\subsection{Cover Matrix} Our first step, and also a key to later development,
is to adopt a matrix representation of a cover. We use the following
representation from \cite{chen2012overlap}.

\begin{definition}[Matrix Representation of a Cover]
A matrix $ Y \in \mathbb{R}^{n\times n}$ represents a cover $ \mathcal{C} = \{C_k \} $ if $
Y_{ij} = \left|\{C \in \mathcal{C} : i\in C, j\in C\}\right| $. That is, $ Y_{ij} $ equals the
number of clusters that include both node $ i $ and $ j $.
\end{definition}
Each cover has a unique matrix
representation. \yudong{here I do not claim anything for the other direction.} To see this, let us introduce the notion of a cluster assignment
matrix.
\begin{definition}[Cluster Assignment Matrix] $ U \in \mathbb{R}^{n\times r}$ is
the cluster assignment matrix of a cover $ \mathcal{C} = \{C_k, k=1,\ldots,r\} $
 if $ U_{ik} = 1 $ when $ i \in C_k $ and zero otherwise.
\end{definition} 

The cluster assignment matrix $ U $ is another representation of a cover which
shows the clusters that each node belongs to.
Clearly each cover corresponds to a unique $ U $, and each $ U $ corresponds to a unique $ Y $ via the factorization $ Y = UU^\top $ (the $ (i,j) $-th entry of the matrix $ UU^\top $ is the inner product of the $ i $-th and $ j $-th rows of $ U $, which, due to the structure of $ U $, equals the number of shared clusters of node $ i $ and $ j $, i.e., $ Y_{ij} $). In the sequel we will mainly use $ Y $ as
the optimization variable, but the factorization is useful
later for post-processing.

Another way to view the cover matrix $ Y $ is that it assigns to each pair of nodes $(i,j)$ a ``similarity level'' $ Y_{ij} $, measured by the number of shared clusters
between $ i $ and $ j $ \cite{chen2012overlap}. When there is no overlap, the
assigned similarity level is either $ 1 $ ($ i,j $ assigned to the same cluster)
or $ 0 $ (assigned to different clusters). When overlaps are allowed,
nodes sharing many clusters are considered more similar. In contrast, the
network adjacency matrix $ A $ can be viewed as the \emph{observed} similarity
level. With this in mind, we can think of the general objective of OTC detection as: find a series of covers $ Y^t $ such that the assigned similarity level is closed to the observed one at each $ t $, and
the covers change smoothly over time. In general the number of clusters that include both node $i $ and $ j $ might be greater than $ 1 $, so the assigned similarity is also above $ 1 $.

\subsection{Overlapping Temporal Community Detection}

We now give a precise formulation of the above general objective. We adopt an
optimization-based approach to OTC Detection. In particular, we
consider the following framework:
\begin{eqnarray}
\max_{\{Y^t\}} &  & \sum_{t=1}^{T} f_{A^t}(Y^{t}) \label{eq:opt}\\
\textrm{s.t.} &  & \sum_{t=1}^{T-1 }d_{A^{t+1},A^t}(Y^{t+1},Y^{t}) \le\delta, \nonumber \\
 &  & Y^{t}\textrm{ represents a cover}, t=1,\ldots,T. \nonumber
\end{eqnarray}
Here $f_{A^t}(Y^{t})$ is the snapshot quality, which serves two purposes: 1) it
measures how well the cover $Y^{t}$ reflects the network $A^{t}$, i.e., the
closeness between the assigned similarity level encoded in $ Y^t $ and the
observed similarity level in $ A^t $, and 2) it prevents the algorithm from
over-fitting, e.g., generating duplicate communities or many small communities
overlap with each other. The function $d_{A^{t+1},A^t}(Y^{t+1},Y^{t})$ is a
distance function that measures the difference between the covers at time $t+1$
and $t$. Consequently, the first constraint in the above formulation ensures
that the covers evolve smoothly over time. This constraint prefers the
evolutionary path with fewer changes and reflects the inertia inherent in
evolution of groups in real life networks. 
%In the non-overlap case where $ Y^t $
%is restricted to be a partition, the above formulation is similar to
%Evolutionary Clustering \cite{evolutionaryclustering}. 

In this paper we focus on concave $f$ and convex $d$ (w.r.t. $\{Y^t\}$). This
covers many existing methods for clustering without overlap. For example, $f$
can be the modularity function~\cite{Newman2006} $ f_A(Y) = \sum_{i,j} \left(
A_{ij} - \frac{k_i k_j}{2M}\right) Y_{ij} $ (here $ k_i $ is the degree of node
$i$ in $ A $,  $ M $ is the total number of edges, and we ignore the
pre-constant) or the correlation clustering \cite{bansal2004correlation}
objective $ f_A(Y) = -\Vert A-Y \Vert_1 $ (here $ \Vert X \Vert_1 = \sum_{ij}
\vert X\vert $ is called the matrix $ \ell_1 $ norm of $ X $), and $d$ can be the
estrangement \cite{estrangement} $ d_{A^{t+1},A^t}(Y^{t+1}, Y^t) = \sum_{ij} A^{t+1}_{ij}
A^t_{ij} \max\left(Y^t_{ij} - Y^{t+1}_{ij},0 \right)$.

For OTC detection, the difficulty lies in defining quality and distance
functions that can handle overlaps. We propose two novel metrics that are suitable to
this task. For the snapshot quality $ f $, we use the weighted $ \ell_1 $
distance between the cover matrix $ Y $ and the adjacency matrix $ A$: 
\begin{equation*}
 f_A(Y)  =  - \sum_{i,j} \vert C_{ij}(Y_{ij} - A_{ij}) \vert, 
\end{equation*} 
where $C_{ij}$ are some weights to be chosen. In this paper, we use 
the weights $ C_{ij} = \left\vert A_{ij}-\frac{k_i k_j}{2M}\right\vert$, where 
$ k_i $ and $ M $ are defined in the last paragraph. This qualify function generalizes the
correlation clustering objective \cite{bansal2004correlation} and is closely
related to the widely-used modularity quality function \cite{Newman2006} when
there is no overlap. In particular, it penalizes three types of
``errors'' (recall $ Y_{ij} $ is the number of clusters including both $
i $ and $ j $, or the assigned similarity level between $ i,j $): 
\begin{itemize}
\item $ A_{ij}=1$ and $Y_{ij}=0 $: nodes $ i $ and $ j $ are connected but they are assigned to different clusters 
\item $ A_{ij}=0$ and $Y_{ij}\ge 1$: nodes $ i $ and $j $ are disconnected but they share at least one clusters, i.e., the assigned similarity level is positive while the observed one is zero.
\item $ A_{ij}=1$ and $Y_{ij} > 1 $: nodes $ i $ and $ j $ are connected but they share more than one clusters, i.e., the assigned similarity level is higher than the observed one. 
\end{itemize}
Note that in the last two cases, the more clusters $ i $ and $ j $
share, the higher the cost is. This prevents the algorithm from over-fitting by
generating many small clusters with lots of overlap.

For the temporal distance $ d $, we use:
\begin{equation*} 
d_{A^{t+1},A^t}(Y^{t+1},Y^t)  =  \sum_{i,j} A^{t+1}_{ij} A^{t}_{ij} \vert
Y^{t+1}_{ij} - Y^t_{ij} \vert.  
\end{equation*} 
In other words, we measure the change in the assigned similarity level between node $
i $ and $ j $ (i.e., the number of clusters that include both nodes), but only when
there is an edge between $ i $ and $ j $ in both snapshots $ t $ and $ t + 1$. For non-overlapping
clusters, this reduces to the number of persisting edges that change ``state''
from intra-cluster to inter-cluster and vice-versa. Our measure is a modification and 
generalization of the estrangement measure in~\cite{estrangement} to overlapping clusters.

\subsection{Convex Relaxation}

The optimization problem \eqref{eq:opt} is combinatorial due to the constraint
``$Y^{t}$ represents a cover''. Exhaustive search is impossible because there are 
exponentially many possible covers. One option is to use greedy
local search, which a popular choice for optimizing modularity and other clustering
objectives, but it only converges to local minimums and provides no guarantees. 

In this paper, we use convex optimization. There are two advantages
of this approach: 1) it leads to an optimization problem that is efficiently
solvable and guaranteed to converge to the global optimum, and 2) it is possible
to obtain {\it a priori} characterization of the optimal solution (see
Section~\ref{sec:theory}), which provides interesting insights into the problem. To this end,
we relax the cover constraint and solve the following optimization problem:
\begin{eqnarray}
\max &  & \sum_{t=1}^{T} f_{A^t}(Y^{t})\label{eq:cvx_opt}\\
\textrm{s.t.} &  & \sum_{t=1}^{T-1} d_{A^{t+1,A^t}}(Y^{t+1},Y^{t})\le\delta,\nonumber \\
 &  & \left\Vert Y^{t}\right\Vert _{*}\le B, t=1,\ldots,T;\nonumber 
\end{eqnarray}
here $\left\Vert Y^{t}\right\Vert _{*}$ is the so-called trace norm, the sum of
singular values of $Y^{t}$. It is known that the trace norm constraint $
\left\Vert Y \right\Vert _{*}\le B $ is a convex relaxation of the original
cover constraint \cite{chen2012overlap}. We briefly
explain the reason here. Recall that a cover matrix admits the factorization $
Y=UU^\top $, so a cover $ Y $ is positive semidefinite and satisfies
\begin{equation}
 \Vert Y \Vert_* = \sum_{i} Y_{ii} = \sum_i \#(\mbox{clusters that include node } i).
\label{eq:trace_norm}
\end{equation} 
In particular, the right hand side in \eqref{eq:trace_norm} equals $ n$ if $ Y $ represents a partition.
Therefore, as long as $ B $ is no smaller than the right hand side in
\eqref{eq:trace_norm}, then a cover matrix $ Y $ also satisfies the new
constraint, which is thus a relaxation. Although the right hand side in
\eqref{eq:trace_norm} is unknown a priori, in practice we find that choosing $ B
$ to be suitably large, such as $10n$ as is done in our experiment section,
works well. Moreover, the constraint $ \left\Vert Y \right\Vert _{*}\le B $
effectively imposes an upper bound on the amount of overlap and prevents the
algorithm from producing a large number of clusters, which is desirable on its
own right.

Trace norm is known to be a good relaxation for partition matrices both in
theory and in practice \cite{ames2011clique,chen2012sparseclustering}. All partition
matrices with a small number of partitions (which is the case of interest) are
low-rank, and trace norm is the tightest convex relaxation of low-rank matrices
in a formal sense \cite{recht2010guaranteed}. Moreover, trace norm utilizes the
graph eigen-spectrum which has long been known to reveal hidden clustering structure and is the basis of the highly successful spectral clustering methods. This advantage of trace norm carries over to overlapping clusters \cite{chen2012overlap}.
With this relaxation and our choice of $ f $ and $ d $,
\eqref{eq:cvx_opt} becomes a convex program and can be solved in polynomial time
using general-purpose convex optimization packages such as SDPT3. In Appendix \ref{sec:fast_algo}, we
describe a specialized gradient descent algorithm, which is even faster. 

\subsection{Post-processing}
\label{sec:postpro}

Ideally, the optimal solution $\hat{Y}^{t}$ would represent a cover, which could
be easily extracted from $\hat{Y}^{t}$ (e.g. by finding all maximal cliques); in
the next section we provide one sufficient condition for this to happen. In
practice, however, because of the relaxation, $\hat{Y}^{t}$ may not have the structure of a cover matrix. But it is empirically
observed that $ \hat{Y}^t $ is usually \emph{close} to a cover matrix; in
particular, the optimization can be viewed as a ``denoising'' procedure, which
filters out most (though not all) of the noise in the observation $ A^t $ and
makes the underlying clustering structure more clear. Therefore, a good
clustering is likely to be extracted from $ \hat{Y}^t $ via simple
post-processing. We describe one such procedure below.

Recall again that a cover matrix can be factorized as $ Y = U U^\top $, where $
U $ is an assignment matrix of non-negative entries, with $ U_{ik} =1$
indicating node $ i $ in cluster $ k $. Therefore, performing Non-negative
Matrix Factorization (NMF) \cite{seung2001nmf} on a cover $ Y $ gives the
corresponding clustering assignment. When the optimal solution $ \hat{Y} $ is
not an cover but close to be one, we expect that performing NMF $ \hat{Y} =
\hat{U}\hat{U}^\top$ would still produce an approximate assignment matrix $
\hat{U} $, which is then rounded to be an exact assignment matrix. In
particular, we declare node $ i $ to belong to cluster $ k $ at time $ t $ if $ \hat{U}^t_{ik}
\ge 0.5$.

\subsection{Remarks on Our Method}\label{sec:algorithm_remark}
\textbf{Mapping communities:}  Practical application sometimes requires the
communities at time $t $ to be mapped to those at $ t-1 $, in order to track the
evolution of communities. In the experiment section, we use the mapping method
in \cite{estrangement}, which still works when $ Y $ is a cover instead of a
partition. The method involves mapping those communities across consecutive snapshots that have
the maximal mutual Jaccard overlap between their constituent node-sets, and
generating new community identifiers only when needed. 

\textbf{Online algorithm:} In some cases it might be interesting to use an 
online version of the algorithm \eqref{eq:cvx_opt}: At each time $ t $ when a
new snapshot $ A^t $ becomes available, we obtain a new cover $ Y^t $ by solving
the following problem:
\begin{eqnarray}
\max_{Y^t} &  &  f_{A^t}(Y^{t})\label{eq:cvx_opt_online}\\
\textrm{s.t.} &  & d_{A^t,A^{t-1}}(Y^{t},Y^{t-1})\le \delta^t, 
\quad \left\Vert Y^{t}\right\Vert _{*}\le B,\nonumber 
\end{eqnarray}
where $Y^{t-1} $ is the solution from the last snapshot $ t-1 $ and is considered
fixed. Rigorously speaking, the solution to the online formulation is in general
different from that to the offline one. But we expect in practice the online
formulation will perform reasonably well, and various updating rules can be
adopted to choose the online upper bound $ \delta^t $. We do not delve into this in
this paper.

\textbf{Complexity and Scalability:} Using the fast gradient descent algorithm, the space 
and time complexities of our method both scale linearly with the problem size (the numbers of nodes, 
edges and snapshots); see Appendix \ref{sec:complexity} for details. With the online 
implementation suggested above, the dependence on the snapshot number can be further alleviated. 
Our method is therefore amenable to large datasets.

\section{Theoretical Analysis}

\label{sec:theory}

In this section we provide theoretical analysis on the performance of our
algorithm. In particular, our analysis shows that if the adjacency matrices
$A^{t}$ are generated from an underlying persistent partition according to a
generative model, then with high probability our method will recover the
underlying partition as long as $K=\Omega(\sqrt{n/m})$, where $K$ is the minimum
cluster size in the partition and $m$ is the number of snapshots it persists for. This highlights the benefit of temporal
clustering: a small cluster of size $ \sqrt{n/m} $ is likely to be undetectable
if each snapshot is considered individually (e.g., the cluster might not  
be connected in each single snapshot), but can be recovered by temporal
clustering if the cluster persists for $ m $ snapshots and all snapshots are
used. This result is quite revealing: traditional single-snapshot clustering
techniques can only find clusters that are large in size, but temporal
clustering is capable of detecting clusters that are small in size but large in
the time axis. Moreover, our theorem predicts a specific tradeoff between the
``spatial size'' $ K $ and the ``temporal size'' $ m $: with four times more
temporal snapshots, one can detect a cluster that is half as small
spatially. We believe this is the first such result in the literature.

We now present our theorem. We use a generative model which can be considered as
a multi-snapshot version of the classical and widely studied planted partition
model (a.k.a. stochastic block model) \cite{condon2001algorithms}.

\begin{definition}
[Multi-Snapshot Planted Partition Model]
Suppose $n$ nodes
are in $r$ disjoint clusters, each with size $K$, and this clustering
structure does not change over time (see remarks after the theorem). Let $ Y^* $ be the matrix that represents this clustering.
The adjacency matrices $A^{t},t=1,\ldots,m$
are generated as follows: if node $i$ and $j$ are in different clusters,
then there is an edge between them (i.e. $A_{ij}=1$) with
probability $q$, independent with all others; if they are in the same cluster, 
then $A_{ij}=1$ with probability $p$.
We assume $q<\frac{1}{2}<p$ are constants independent of $n$, $m$
and $K$.
\end{definition}
Since the underlying partition does not change,
we impose the constraint $\sum_{t=1}^{T-1} d_{A^{t+1},A^t} (Y^{t+1}, Y^{t})=0$, which is equivalent to $ Y^t = Y, \forall t $. Rewritten in an equivalent minimization form, our algorithm becomes 
\begin{eqnarray}
\min_{Y} &  & \sum_{t}\sum_{ij} C_{ij} \left|Y_{ij}-A_{ij}^{t}\right| \label{eq:no_jump_opt}\\
s.t. &  & \left\Vert Y\right\Vert _{_{*}}\le{n}.\nonumber
\end{eqnarray}
Note that under the multi-snapshot planted partition model, we have $ C_{ij} = \left\vert A_{ij}-\frac{k_i k_j}{2M}\right\vert \approx |A_{ij}-s|$, where $ s :=  p\frac{K}{n} + q\left(1-\frac{K}{n}\right) \in (q,p)$. The following theorem characterizes when (\ref{eq:no_jump_opt}) recovers true underlying partition matrix $Y^{*}$. 
\begin{theorem}
\label{thm:no_jump}Suppose $ C_{ij} = |A_{ij}-s| $. Under the multi-snapshot planted partition model,
if $K= \Omega( \sqrt{\frac{n}{m}} )$,
then $Y^{*}$ is the unique optimal solution to the convex optimization problem
(\ref{eq:no_jump_opt}) with probability converging to $1$ as $n\rightarrow\infty$.
\end{theorem}
\iftechreport
The proof is given in Appendix \ref{sec:proof_of_theorem}.
\else
\begin{proof}
  The proof is in~\cite{bbntechreport} due to space constraints.  
\end{proof}
\fi

\textbf{Remark on Theorem \ref{thm:no_jump}:} Although the multi-snapshot
planted partition model assumes that the underlying clustering structure does
not change, and that the clusters do not overlap, we conjecture similar
theoretical guarantees can be obtained with these restrictions removed. 
In particular, we expect that our algorithm can detect clusters of size $\Theta(\sqrt{\frac{n}{m}})$ 
even if the underlying structure changes, provided that between consecutive changes there 
are at least $m$ snapshots. This
conjecture is supported by the experimental results in section \ref{sec:expt}.

%% file: expt.tex
\section{Experimental Results}
\label{sec:expt}

We apply our method to two synthetic datasets and three real-world datasets. Our
synthetic networks are random graphs generated according to an underlying
community structure evolution. Each snapshot is an instantiation of a random
graph generated by connecting each pair
of nodes sharing at least one community with probability 0.5, and with
probability 0.2 otherwise (including the case where one or both of the nodes are
not in any community). Note that we allow some nodes in some snapshots to not
belong to any community, as is often true in real scenarios. Also, note that
nodes sharing more than one communities are not connected with a higher
probability. This makes overlapping communities harder to detect and is a better
test of the detection methods. 

Using this prescription, we generate two synthetic time-varying networks to
validate our method and demonstrate its efficacy. We compare the results
obtained with and without overlap allowed, and with and without the smoothness
constraint. A popular temporal clustering method using multi-slice modularity \cite{Mucha2010} is also considered. 
%There are no available methods and
%implementations for detecting the OTC structure in literature.

\iftechreport
The four real network datasets considered in this section include MANET, international trade, AS links, and the MIT Reality Mining Data.
\else
The three real network datasets considered in this section include MANET, international trade, and AS links. 
We provide additional results on another real dataset (the MIT Reality Mining Data) in the technical report \cite{bbntechreport}.
\fi

\vikas{Some other things we can do here:
  \begin{itemize}
  \item Compare running times of SDP vs factorization
  \end{itemize}
}
\subsection{Synthetic Random Networks I}

In the first synthetic experiment, we demonstrate the advantage of
considering the temporal aspect and allowing overlap, and that there is
clustering structure that can be detected only if we consider both. We generate
the network snapshots as follows. Suppose there are 120 nodes and 5 underlying
communities. Community 1 is a small 15-node group including nodes 0 through 14.
Community 2 and 3, both of size 38, consist of nodes 15-52 and 47-85,
respectively, and overlap at 5 nodes (47-52). Communities 4 and 5, both of smaller
size 20, include node 85-104 and 100-119, respectively, and overlap at 5 nodes
(100-104). 

Since community 1 is small, in light of Theorem \ref{thm:no_jump}, we expect
that single-snapshot methods are unable to detect it due to noise/randomness in
the network, but temporal methods will find them. Community 2 and 3 are large
but overlap with each other, so only methods that allow overlap
would detect them, even if the snapshots are considered individually. Finally,
communities 4 and 5 are small and overlapping, and are thus expected to be
discoverable only when both the temporal and overlap aspects are considered.
This is indeed the case in our experiments. The results are shown in Fig
\ref{fig:4comb_all} to \ref{fig:4comb_none}. Visualizing overlapping temporal
communities is not a trivial task. Here we extend the approach used in Fig. \ref{fig:schematic} to allow overlaps, which is explained in the caption of Fig
\ref{fig:4comb_all}. Fig \ref{fig:4comb_all} shows the results of our method,
which nicely detects all the underlying structure. Fig \ref{fig:4comb_overlap}
shows the result of our method but with $ \delta $ set to infinity, so there is no temporal smoothness constraint and snapshots are
considered independently. In this case, communities 1, 4 and 5 are not
recovered completely. Fig. \ref{fig:4comb_temp} shows the result when overlap is not
allowed, i.e., we impose the constraint $ Y^t_{ij} \le 1, \forall t,i,j. $. All
overlapping structure is clearly lost. Fig \ref{fig:4comb_none} shows that
result when $ \delta=\infty $ and overlap is not allowed; one can see a further
degradation of performance.

We also measure the performance of the above four methods by computing the
distance of the recovered community structure from the ground truth. We use the distance
$ \sum_{t=1}^T \Vert Y^{*t} - \hat{Y}^t \Vert_1 $, where $ Y^{*t} $ denotes the
cover matrix of the ground truth, and $ \hat{Y}^t $ the one found by a clustering 
algorithm. The results are given in the second row of Table \ref{tab:distance}.
The error is an order of magnitude smaller when both the overlapping and
temporal aspects are considered. 

\textbf{Comparison with existing schemes}:
Although there has been much work on  community
detection algorithms, almost none allows simultaneously discovering overlapping
and temporal communities.  Thus, we can only compare against some representative
non-overlapping temporal community detection algorithms.  We compare against
the widely cited temporal community detection scheme presented in \cite{Mucha2010}.
This method involves two parameters,
the resolution $ \gamma $ for the modularity function 
and the inter-slice coupling strength $ \omega $. Since the ground truth clustering structure does not change over time,
a large $ \omega $ is used to force a static output. We then search over different values of $ \gamma $ and
use the one that gives the smallest error.
\iftechreport The recovered community structure is  shown in Fig. \ref{fig:4comb_mucha}. \fi
We find that this method cannot identify the overlap structure (as expected), and fails to recover the 
non-overlapping portions of small communities (community 4 and 5\iftechreport \else; see \cite{bbntechreport} for details\fi). The recovery error, given in the last 
column of Table \ref{tab:distance}, is also high. 
%\else We also compare with the classical algorithm in \cite{Mucha2010}, 
%which fails to recover (even the non-overlapping portion of) small communities; see \cite{bbntechreport} for details. 

\begin{figure}[!t]
\centering
\includegraphics[width=1\linewidth]{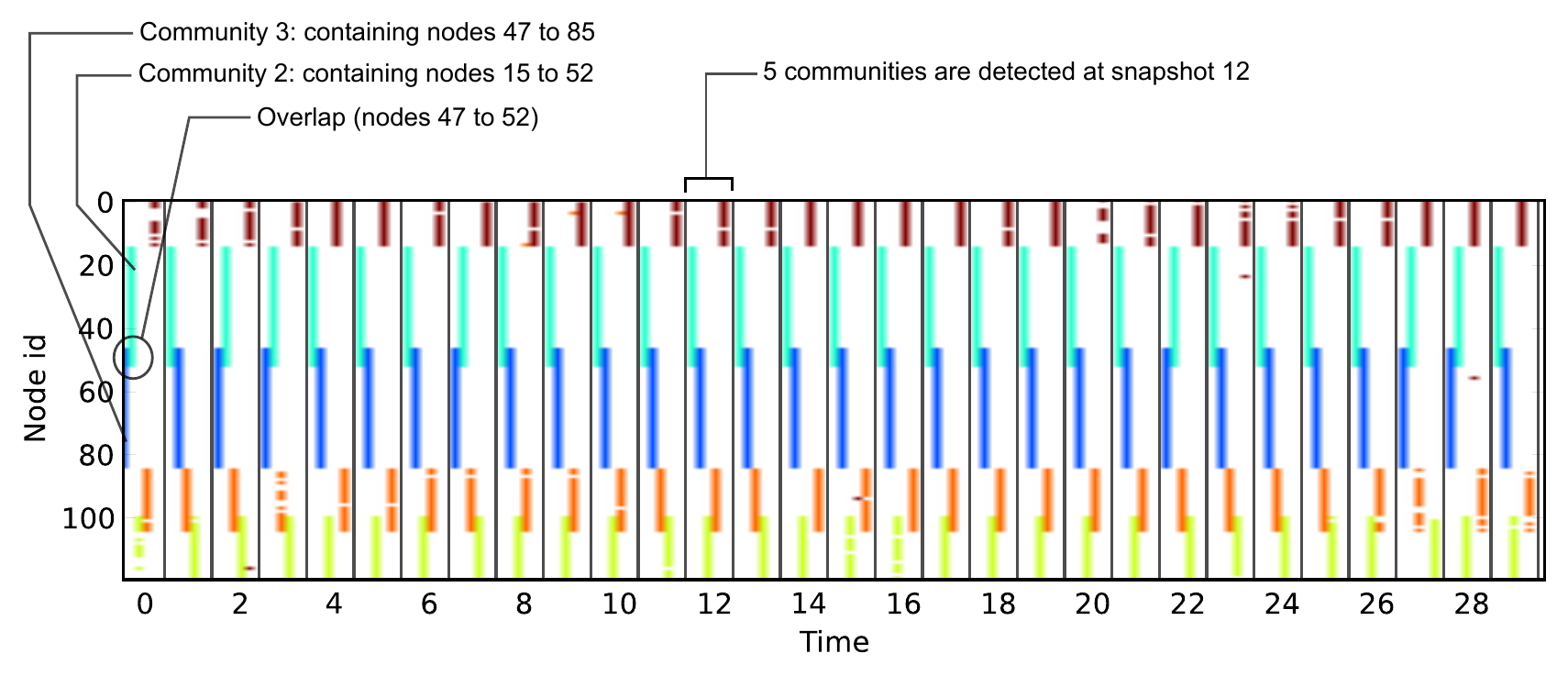}
\caption{Synthetic experiment I: Overlapping Temporal Community structure detected by our algorithm. In the figure, the strips between the consecutive vertical
black lines shows the community assignment for one snapshot. For each snapshot, there are a
number of colored vertical strips, each representing a community which contains
nodes with corresponding labels on the $ Y$ axis. For example, the leftmost blue
strip represents a community (Community 3) at time 0 containing nodes 47 to 85, and the cyan
strip to its right is a community (Community 2) containing nodes 15 to 52; nodes 47 to 52 are in
both communities. Across snapshots, communities with the same color are those
that are mapped to each other using the mapping method described in
Section~\ref{sec:postpro}.
This figure shows that our method faithfully recovers the underlying 5-community
structure that is used to generate the network.}
\label{fig:4comb_all}
\end{figure}

\begin{figure}[!t]
\centering
\includegraphics[width=1\linewidth,height=0.35\linewidth]{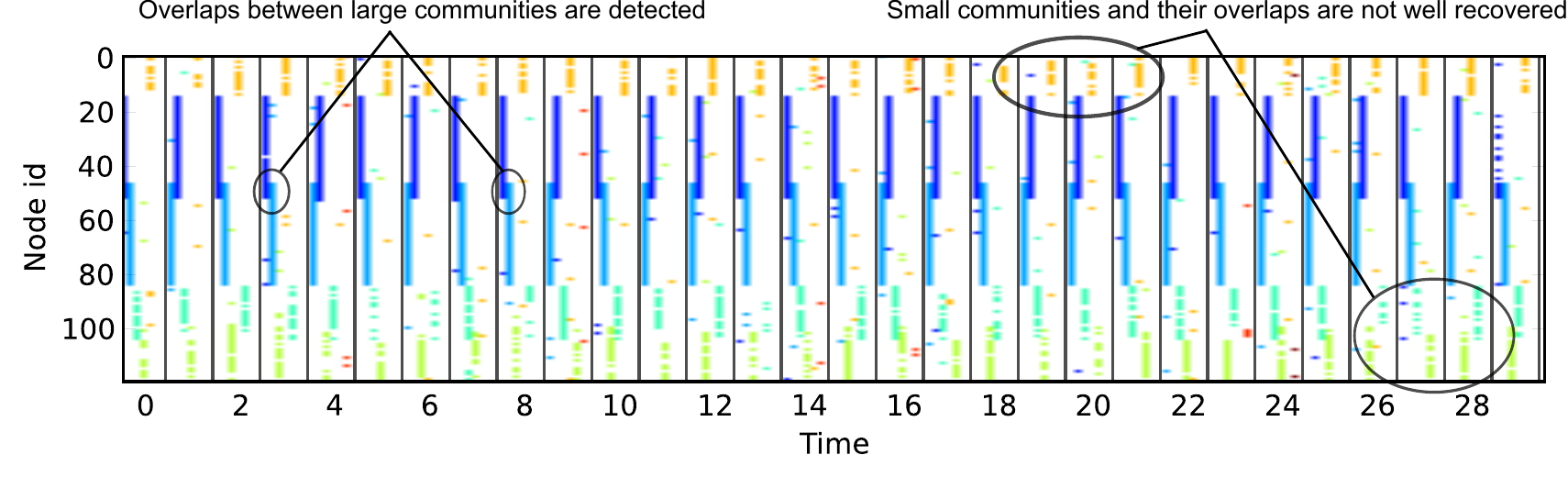}
\caption{Synthetic experiment I:Clustering result when overlaps are allowed but without temporal smoothness constraint. Small communities and their overlaps are not well recovered.}
\label{fig:4comb_overlap}
\end{figure}

\begin{figure}[!t]
\centering
\includegraphics[width=1\linewidth,height=0.35\linewidth]{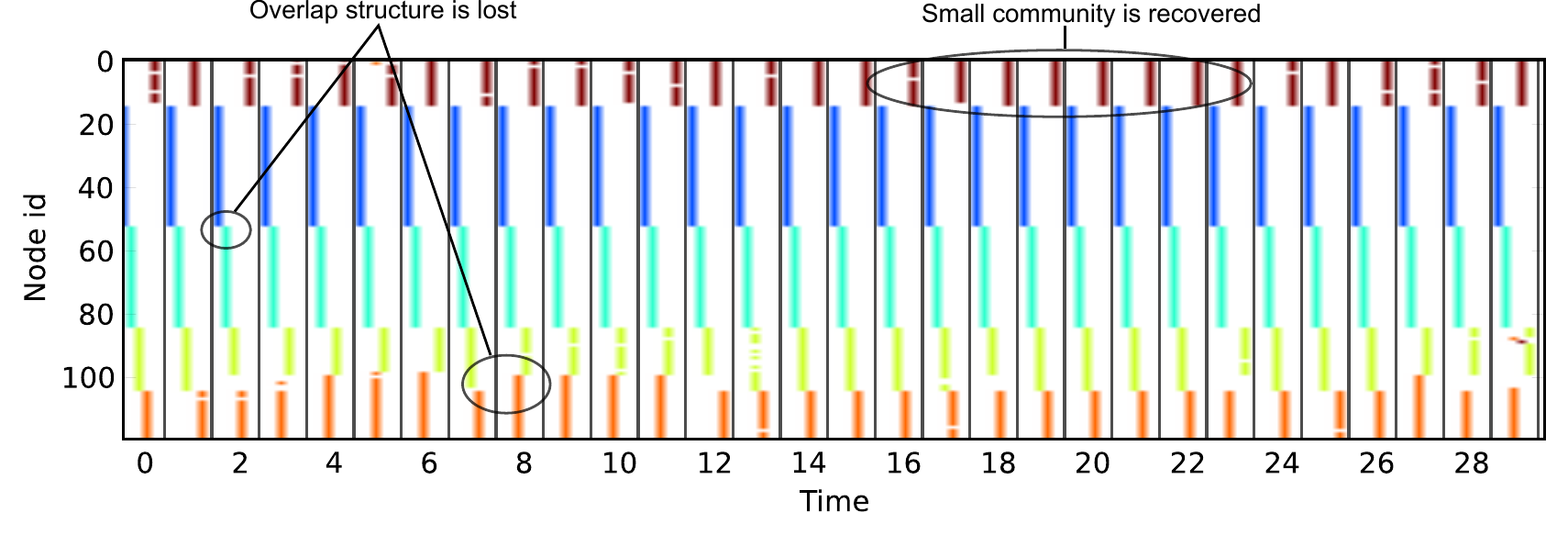}
\caption{Synthetic experiment I: Clustering result with temporal smoothness constraint but not allowing overlaps. The overlap structure is lost.}
\label{fig:4comb_temp}
\end{figure}

\begin{figure}[!t]
\centering
\includegraphics[width=1\linewidth,height=0.35\linewidth]{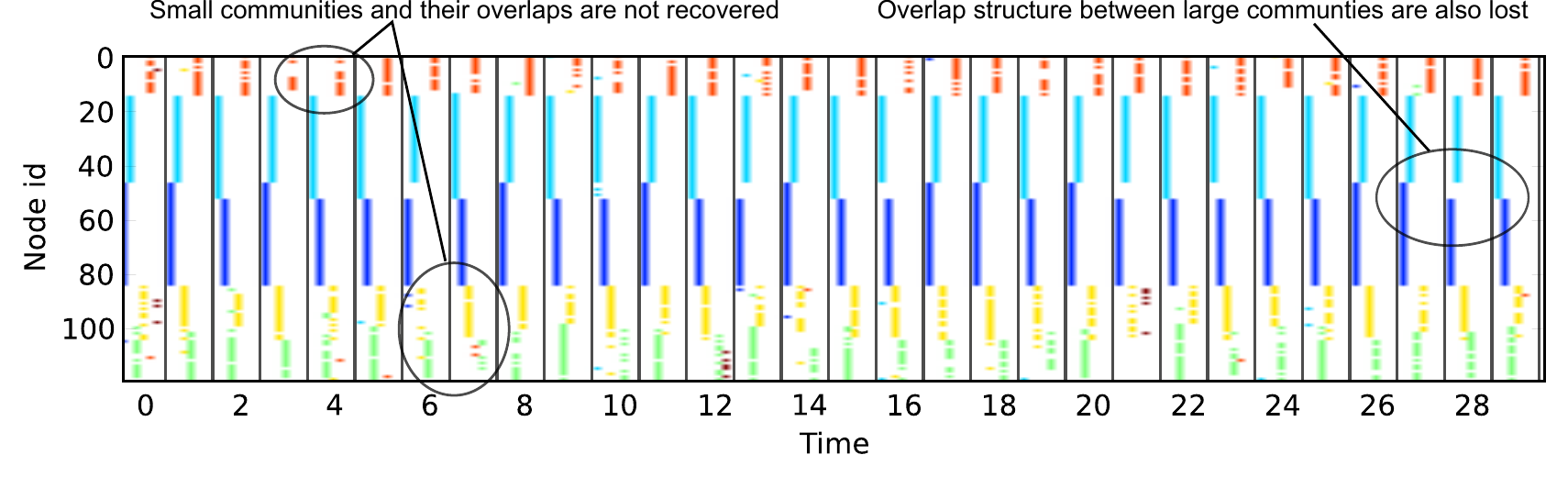}
\caption{Synthetic experiment I: Clustering result without temporal smoothness constraint and no overlaps. Small communities are not well recovered and the overlap structure is lost.}
\label{fig:4comb_none}
\end{figure}

\iftechreport
\begin{figure}[!t]
\centering
\includegraphics[width=1\linewidth,height=0.3\linewidth]{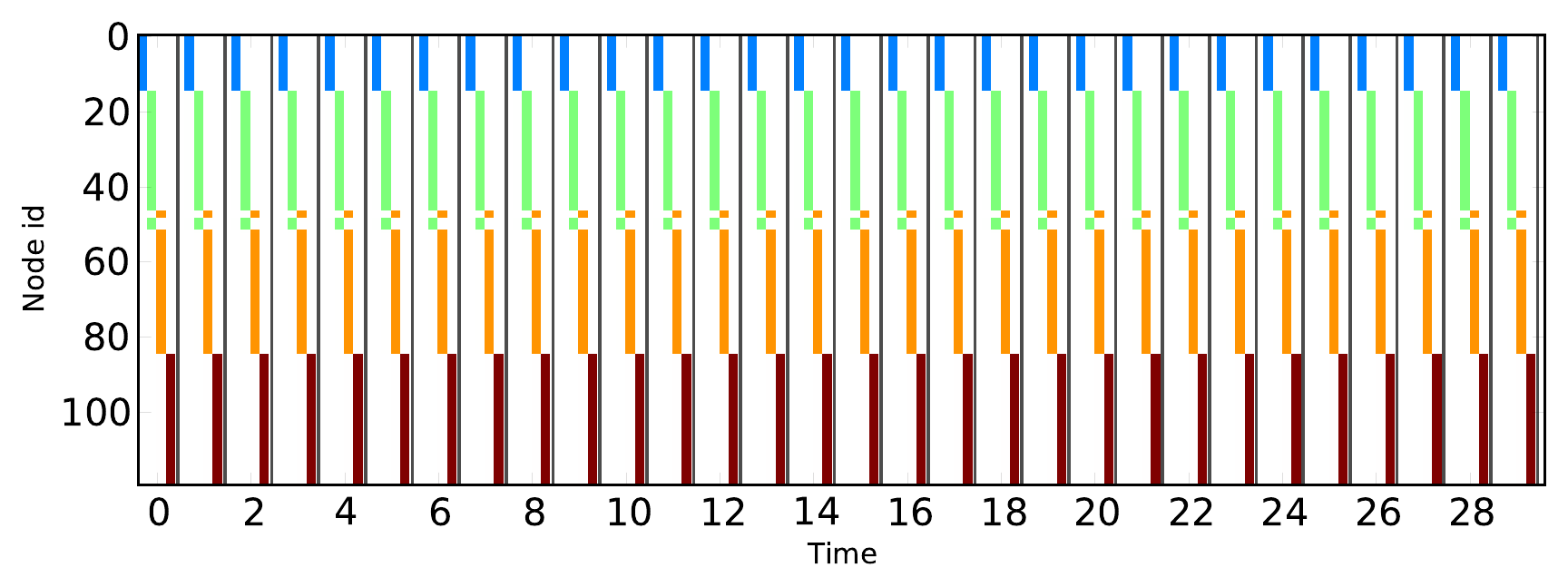}
\caption{Synthetic experiment I: Clustering using the multi-slice modularity method in \cite{Mucha2010}. The two small communities 4 and 5 are incorrectly identified as one cluster, and the overlap structure is lost.}
\label{fig:4comb_mucha}
\end{figure}
\fi

\begin{table}[!t]
%\renewcommand{\arraystretch}{1.3}
% if using array.sty, it might be a good idea to tweak the value of
% \extrarowheight as needed to properly center the text within the cells
\caption{Distance from ground truth for synthetic experiments. }
\label{tab:distance}
\centering
% Some packages, such as MDW tools, offer better commands for making tables
% than the plain LaTeX2e tabular which is used here.
%\begin{tabular}{|c|c|c|c|c|}
\begin{tabular}{|p{1.0cm}|p{1.25cm}|p{0.99cm}|p{1.2cm}|p{0.8cm}|p{0.8cm}|}
\hline
& Overlap+ temporal & Overlap only & Temporal only & None & Ref. \cite{Mucha2010}\\
\hline
\hline
Expt I & 3133 & 27203 & 20646 & 32789 & 27390\\
\hline
Expt II & 1646 & 14843 & 8318 & 12534 & 7450\\
\hline
\end{tabular}
\vspace{-0.2in}
\end{table}

\subsection{Synthetic Random Networks II}

The second synthetic experiment demonstrates the ability of our method to detect
and track time-varying cluster structure, including the overlap, merger,
emergence,
splitting, growth, and shrinking of communities. We describe how we generate the
snapshots. The network  has 100 nodes, and the underlying clustering structure
has four phases, each with 10 snapshots: 

\begin{itemize}
\item Phase I: There are two communities: community 1 includes nodes 0 to 39, and
community 2 includes nodes 40 to 79. The structure does not change during this
phase.
\item Phase II: Community 1 remains the same, but community 2, now with
nodes 30 to 69, overlaps with community 1.  The structure does not change during
this phase.
\item Phase III: Communities 1 and 2 merge into a large community A, which
consists of nodes 0 to 69. This community then gradually shrinks: at each time
there is one node leaving, and at the end of this phase, community A has nodes 0
to 60. On the other hand, there is a new community B, including nodes 75-99,
emerges at the beginning of this phase and remains unchanged throughout this
phase.
\item Phase IV: Community A splits into two smaller ones
consisting of nodes 0-19 and 20-59, respectively. Community B grows by absorbing
nodes 60-74, and thus has nodes 60-100. The structure does not change
in this phase.
\end{itemize}

As can be seen in Fig.~\ref{fig:jump}, our method performs quite
well in recovering the underlying evolving structure. This complements our
theoretical results in section~\ref{sec:theory}, and shows that our method can
handle overlaps and detect jumps in the structure. We compare our method with
those that do not allow overlap, or ignore the temporal aspect; see Table~\ref{tab:distance}.
Our method again outperforms other methods by a
large margin. 

\textbf{Comparison:} 
\iftechreport We also compare with the algorithm in~\cite{Mucha2010}. 
For this algorithm, we search for the best parameters $ (\gamma,\omega)$ that give the 
smallest error. The recovered community structure is shown in Fig.~\ref{fig:jump_mucha},
and the error is given in the last column of Table~\ref{tab:distance}. One observes that 
it cannot detect overlaps. Without considering overlaps, our method is competitive 
with a state-of-the-art algorithm that specializes for temporal clustering.
\else
The result using multi-slice modularity~\cite{Mucha2010} is 
provided in Table~\ref{tab:distance}; see~\cite{bbntechreport} for a plot of the community structure it found. We find that 
our method is competitive with this method in recovering temporal dynamics, and is capable of identifying
additional overlap structures.
\fi 

\begin{figure}[!t]
\centering
\includegraphics[width=1\linewidth,height=0.35\linewidth]{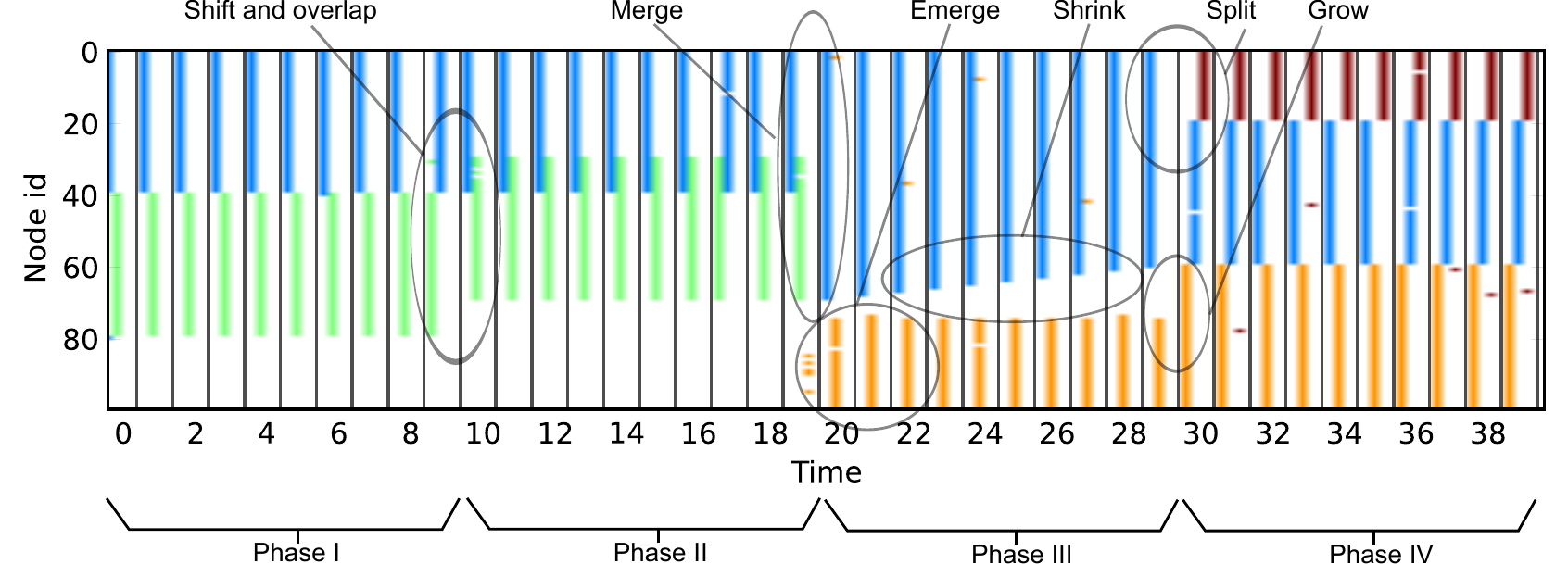}
\caption{Clustering result of our method for synthetic experiment II. Our method is able to detect the merge, emerge, shrink, split, and growth of communities, as well as their overlaps.}
\label{fig:jump}
\end{figure}

\iftechreport
\begin{figure}[!t]
\centering
\includegraphics[width=1\linewidth,height=0.3\linewidth]{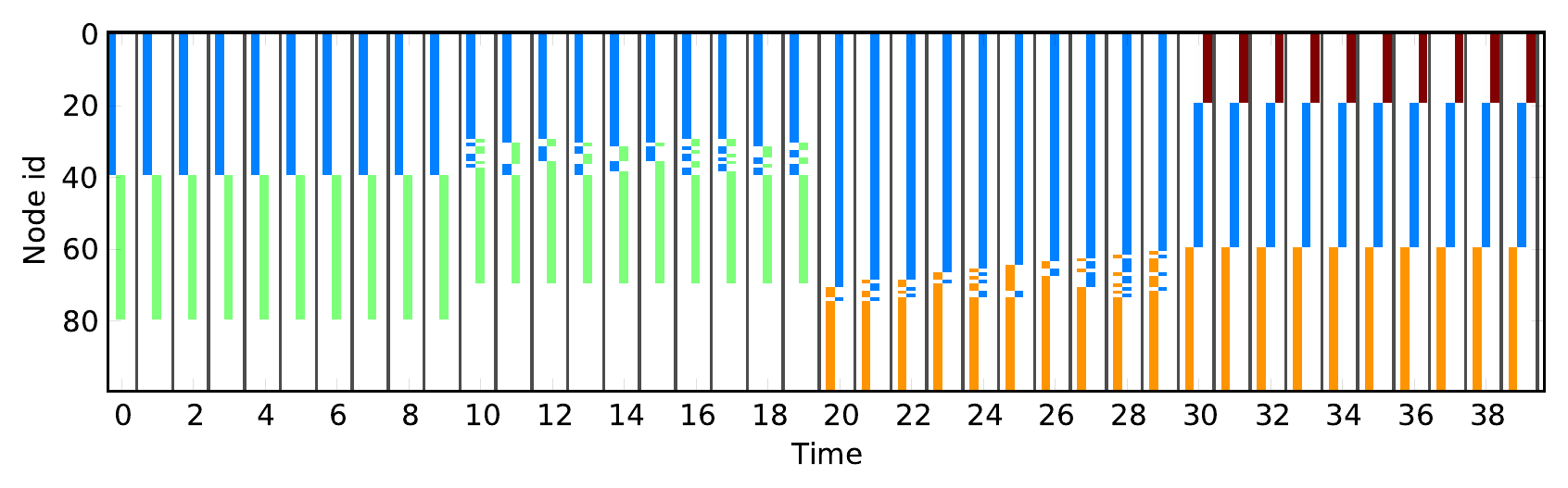}
\caption{Clustering result of the multi-slice modularity method in~\cite{Mucha2010} for synthetic experiment II.}
\label{fig:jump_mucha}
\end{figure}
\fi

\subsection{Real MANET Scenario}

We now present results on a real wireless network with mobility. The data is
based on the mobility trace from an experiment scenario in New Jersey as described
in~\cite{lakehursttrace}. We use a 40 node version of the scenario where the
nodes are organized into three teams.  The teams
move from an initial point to a target point using two primary routes over a
three-hour period. The scenario is divided into several phases, each associated with a
rendezvous point. 
%(marked in Fig.~\ref{fig:lakehurstmap}. 
During each phase the teams
move from one rendezvous point to the next and pause before moving on.  There
are also six \textit{leader} nodes which have high range radios and are mostly
in range of each other.

The input to our algorithm is 711 network snapshots formed by the wireless connection between the nodes. The physical locations of the nodes, as well as the underlying team structure, is unknown to the algorithm. The community structure found by our algorithm is shown in Fig.~\ref{fig:manet}. We find that the leader nodes form a small yet persistent community (shown in orange in Fig.~\ref{fig:manet}), which can only be detected by our clustering
method. We also find that the overlapping temporal community structure is
basically invariant for each phase of scenario even when the topology as well as
the instantaneous community structure without overlap is changing. Thus, we show that in this
case the overlapping temporal community structure detected by our method reveals
a structural pattern that remains invariant even with a fair bit of mobility.

% \begin{figure}[!t]
% \centering
% \includegraphics[width=0.75\linewidth]{lakehurst}
% \caption{Lakehurst scenario overview}
% \label{fig:lakehurstmap}
% \end{figure}

\begin{figure}[!t]
\centering
\includegraphics[width=0.95\linewidth,height=1\linewidth]{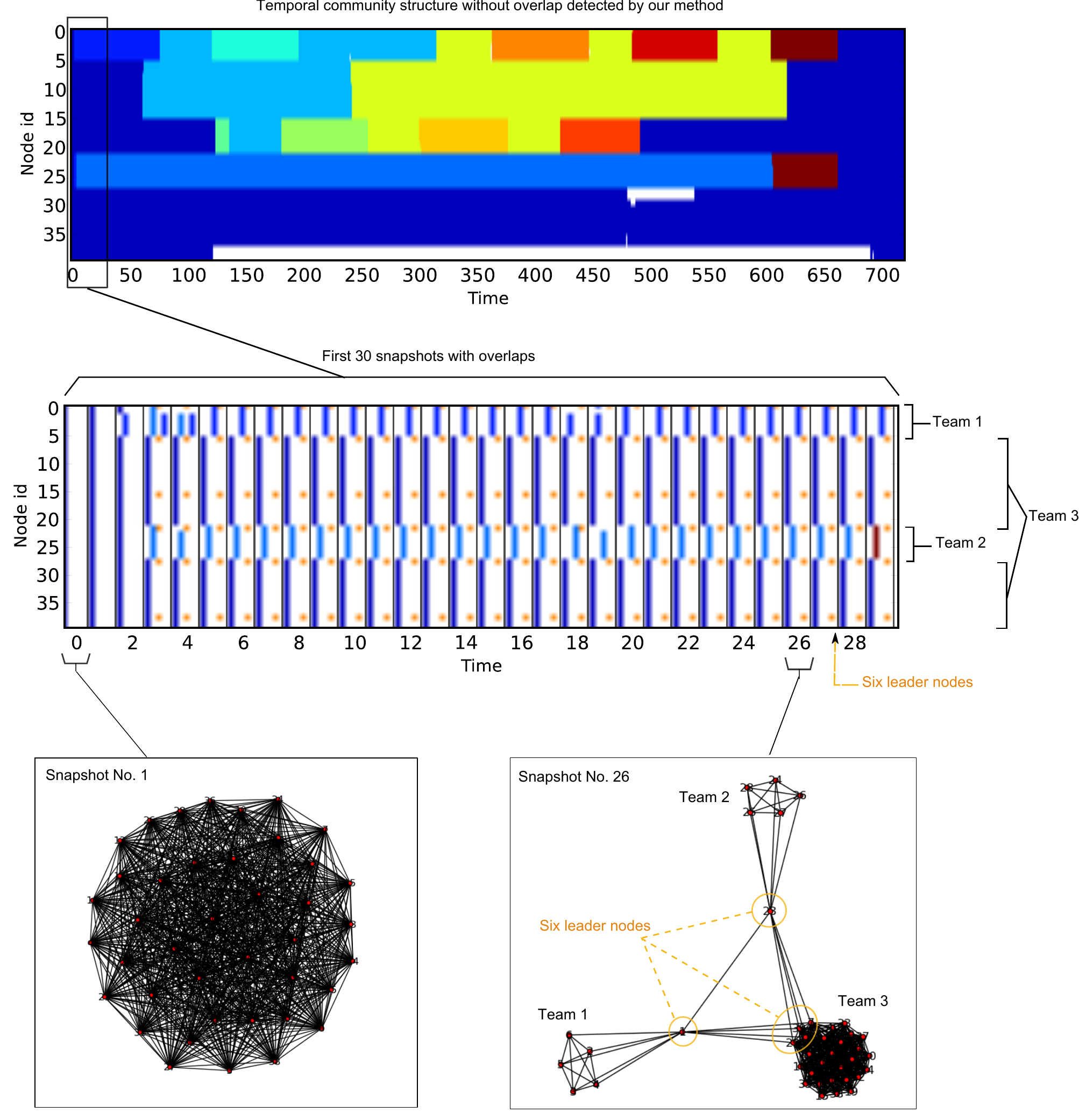}
\caption{Clustering results for MANET data. Top panel: community structure found by our method for
all 711 snapshots, where the colors indicate the community membership of each node at each time; for each node, only the largest cluster it belongs to is shown;
overlaps are not displayed. Middle panel: overlapping community structure for
the first 30 snapshots; 3 teams and the six leader nodes are identified by our
method; the six leader nodes form a small yet persistent community that overlaps
with the other communities; this community can not be detected if overlap is not
allowed (compare with Fig.~\ref{fig:manet_nonoverlap}). Bottom panel: the
observed network structure at two snapshots; at snapshot No.1, all nodes are
densely connected with each other and forms a single community; at snapshot No. 26,
there are three communities corresponding to the three teams; in addition, the
six leader nodes form a community of its own, which is not obvious from looking
at a single snapshot of the network but yet our method is able to detect it.}
\label{fig:manet}
\end{figure}

\begin{figure}[!t]
\centering
\includegraphics[width=0.9\linewidth,height=0.35\linewidth]{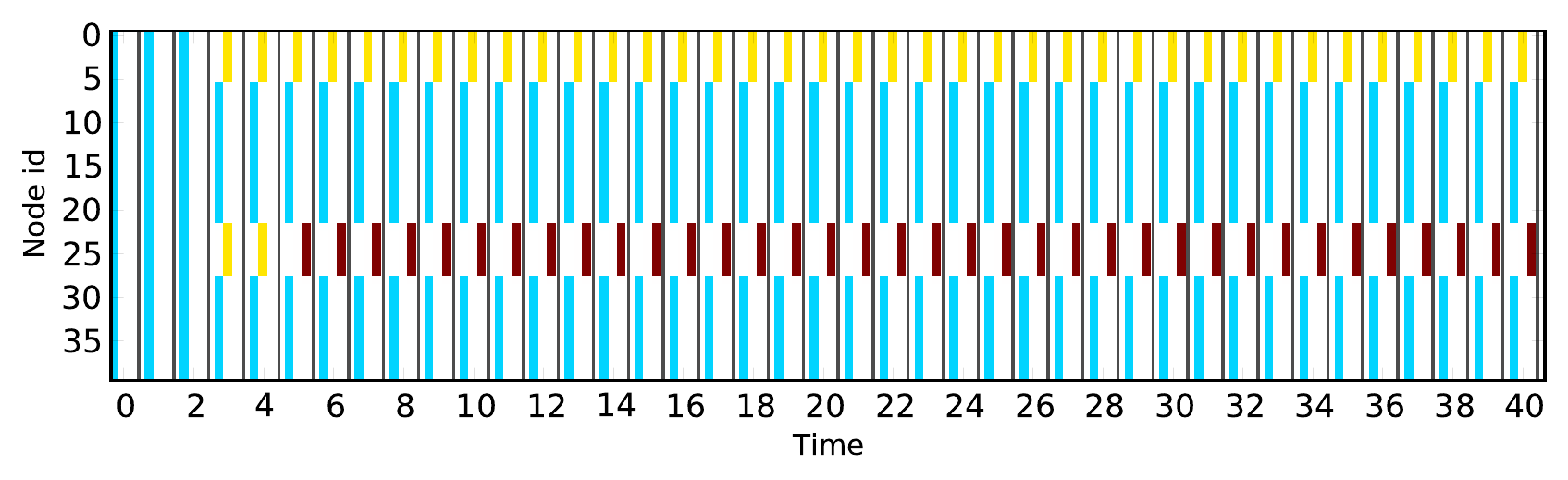}
\caption{Clustering results without overlaps for MANET data for the first 30 snapshots. The community of leader nodes is not detected. }
\label{fig:manet_nonoverlap}
\end{figure}

\iftechreport
\subsection{MIT Reality-mining}

We apply our method to a human-human contact network in the Reality-mining
project~\cite{Eagle}. The results are shown in Fig.~\ref{fig:mit}. Two predominant groups
can be seen, one corresponding to the staff at the MIT Media Lab, and the other corresponding
to the students at the MIT Sloan School of Business. We also observe a discontinuity of the Sloan School community around New Year's break.

\begin{figure}[!t]
\centering
\includegraphics[height=0.5\linewidth]{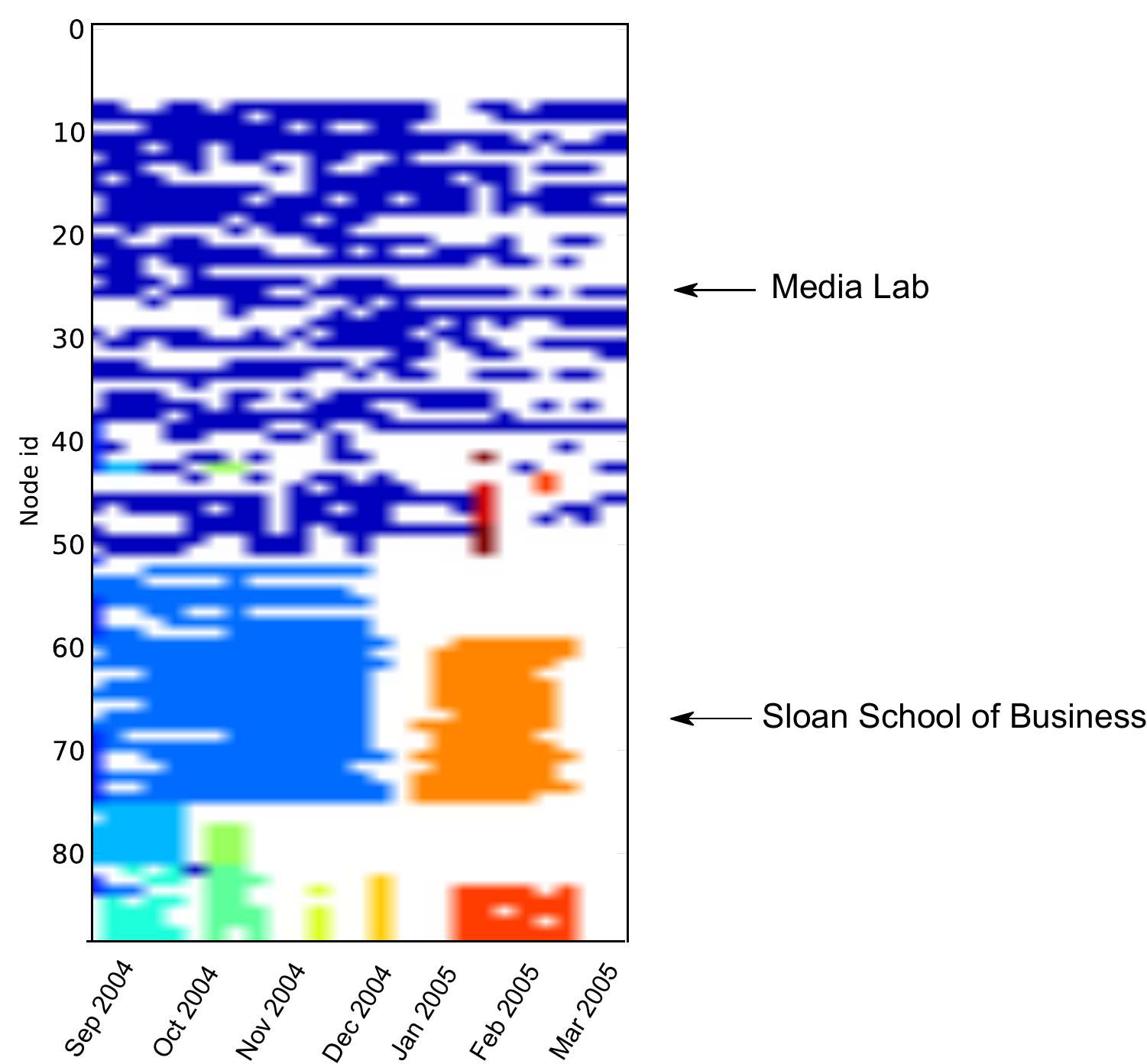}
\caption{MIT Data. No significant overlap is observed, so we only show non-overlapping temporal
community structure found. Two predominant groups can be seen, one corresponding to the staff 
and students at the MIT Media Lab, and the other corresponding to the students at the MIT Sloan 
School of Business.}
\label{fig:mit}
\end{figure}

\fi

\subsection{International Trade Network}

Our next real dataset consists of annual trade volumes between pairs of
countries during 1870--2006~\cite{COWTrade}. We create an
unweighted network each year by placing an edge between two countries if the
trade volume between them exceeds $ 0.1\% $ of the total trade volumes of both
countries; in other words, an edge is drawn if their trade is significant for
\emph{both} of them. This leads to a dynamic network with $ 197 $ nodes and $
137 $ snapshots, which is fed to our algorithm.

Fig. \ref{fig:trade1} shows the post-World War II (1950--2006) community
structure found by our algorithm, where the overlaps are not displayed (for each
node, only the largest cluster it belongs to is shown). Five prominent trade
communities can be immediately identified: Latin-American, US-Euro-Asian,
Ex-USSR Block, West African, and Afro-Asian. One also observes the evolution of
the communities, including the formation of the West African block in 1960
(``the Year of Africa'') due to decolonization, the emergence of the Ex-USSR
block after 1991, as well as Colombia and Venezuela joining the US-Euro-Asian
Block in the 1970s.

More information can be obtained by examining the overlap structure\iftechreport \else(see~\cite{bbntechreport} for a plot)\fi. A number of
countries are associated with multiple communities. For example, US,
Mexico, Colombia and Brazil belong to both US-Euro-Asian and Latin
American blocks. France and Portugal are in the US-Euro-Asian block, but they
both interact with the West African block for a significant number of years.
Similarly, Ivory Coast, Ghana and Nigeria are mainly West African but also
associated with the US-Euro-Asian. Several Asian/Pacific countries,
including Saudi Arabia and Australia, have trade partners in both US-Euro-Asian
and Afro-Asian blocks.
 
\begin{figure}[!t]
\centering
\includegraphics[width=0.7\linewidth]{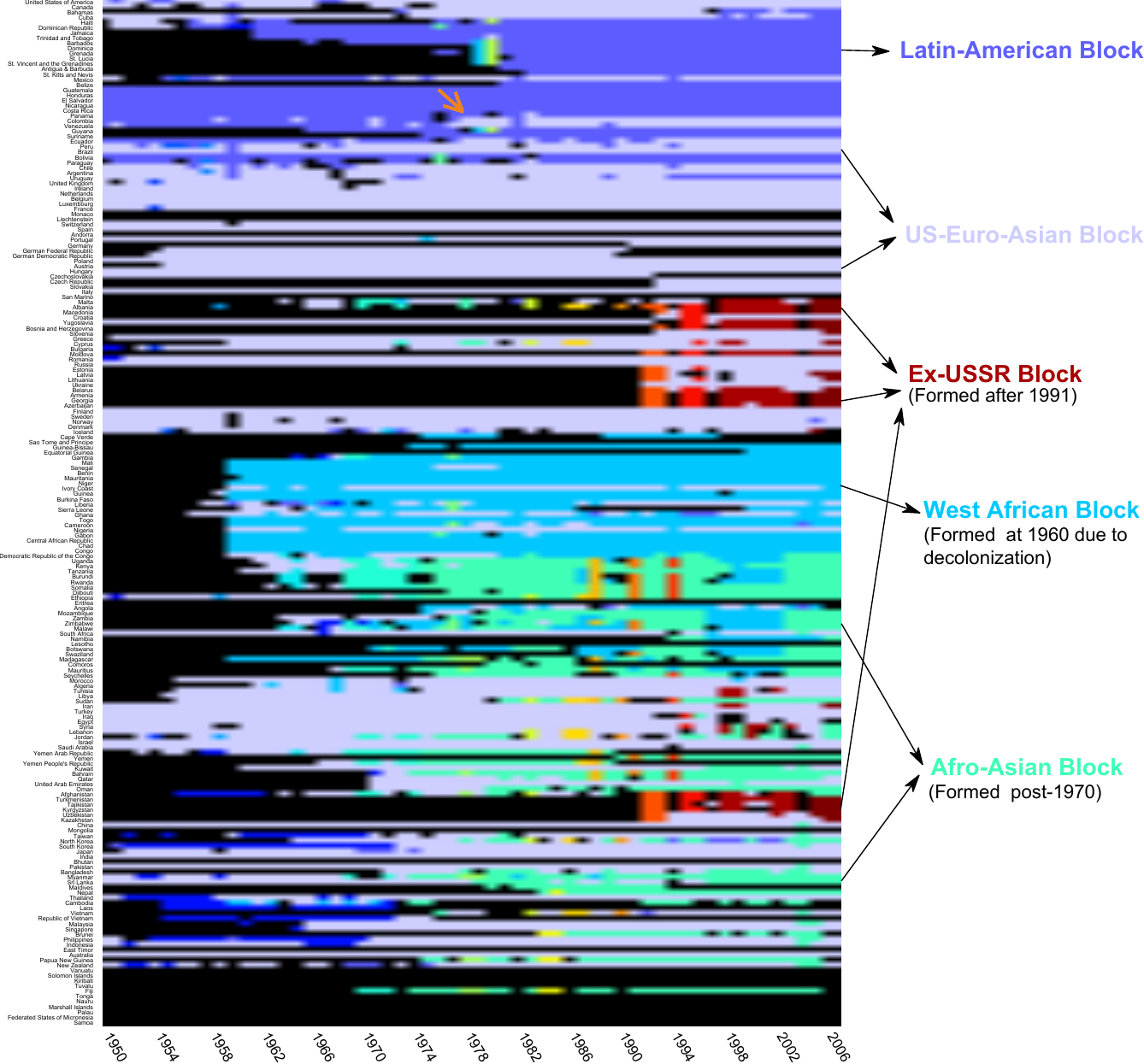}
\caption{Clustering result for the International Trade Network; only years
1950--2006 are shown (overlaps not shown for clarity). Five prominent trade
communities (blocks) can be seen. Moreover, one can observe the emergence of the
Ex-USSR block (after 1992), the West African Block (at 1960) and the Afro-Asian
Block (post-1970), as well as Colombia and Venezuela joining the US-Euro-Asian
block (1970s, orange arrow at the top-middle part of the plot). Note that black
is the background color and not a community.}
\label{fig:trade1}
\vspace{-0.1in}
\end{figure}

\iftechreport
\begin{figure}[!t]
\centering
\includegraphics[height=0.55\linewidth, angle=270]{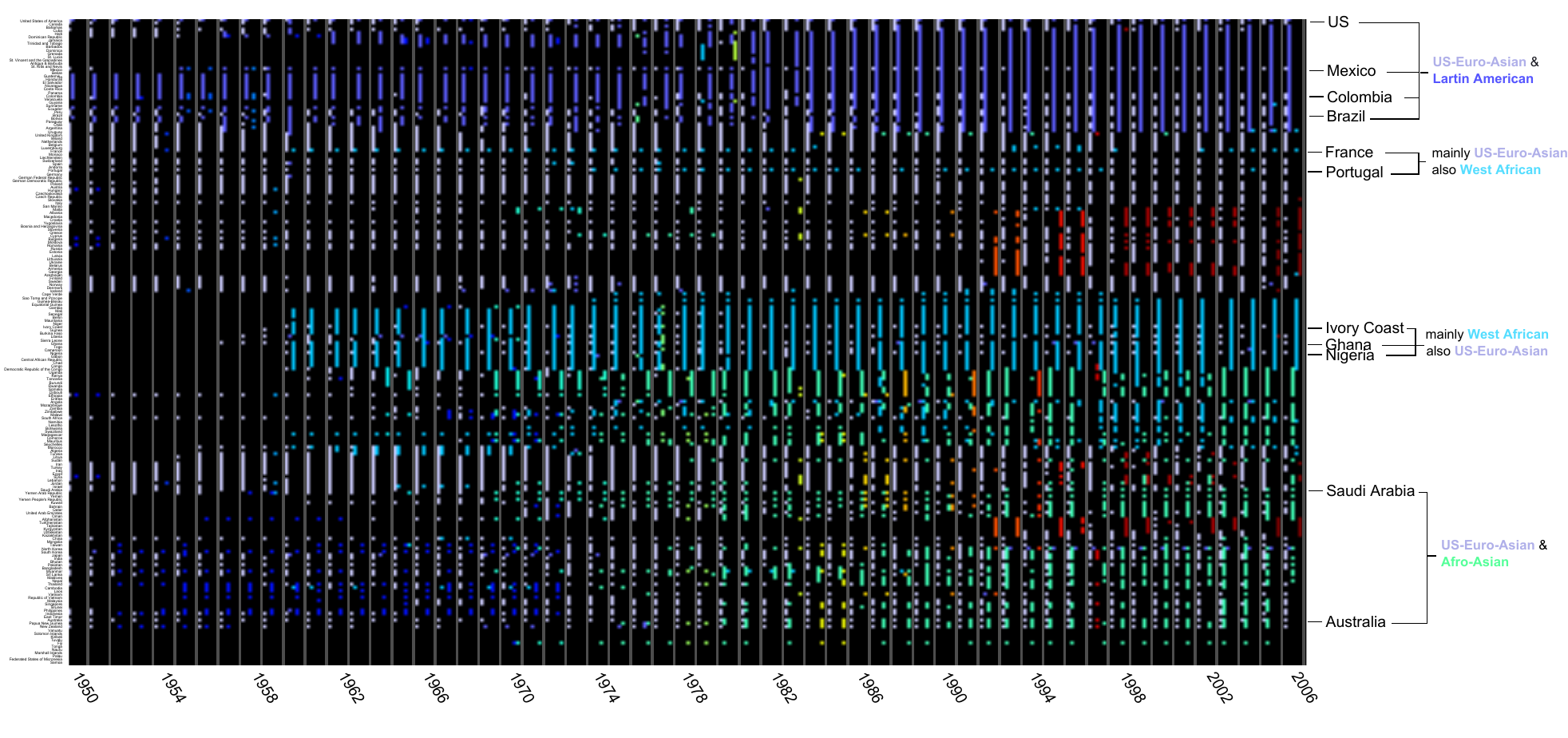}
\caption{Clustering with overlaps for the International Trade Network; only years 1950--2006 are shown. The figure indicates some the countries that are associated with multiple communities. Note that the color black is the background and does not indicate a community.}
\label{fig:trade2}
\end{figure}
\fi

\subsection{The Skitter AS Links Dataset}
Finally, to validate the performance of our algorithm on larger networks, we
analyze the Internet topology at the Autonomous System (AS) level
as collected by CAIDA~\cite{skitter}.
  We obtained quarterly snapshots of the data over an 8 year period starting in
  2000. The data has upto 28000 nodes in some snapshots. Many of
  those are edge nodes with a low degree and do not belong to a community.
 Thus we only consider nodes with degree larger than $ 9 $ in at least
 one snapshot. The final dataset consists of $ 2807 $ nodes and $ 32 $ snapshots.

Among these $ 2807 $ AS's, we identify 90 of them exhibit significant community
structures -- each of them are assigned to community in at least 10 snapshots.
The temporal community structure for these nodes is shown in
Fig.~\ref{fig:skitter1}; overlaps are not shown for clarity.
\iftechreport
Results with overlaps are shown in Fig.~\ref{fig:skitter2}.
\else
See ~\cite{bbntechreport} for the overlap structure.
\fi

%\textbf{Some observations:} 

%\underline{Statistics}: max overlap = 2. number of mapped temporal cluster = 48, longest duration = 18 snapshots. Large cluster size = 57, which is in 2002 first quarter.

We make some initial observations from Fig.~\ref{fig:skitter1}: (1) In upper portion we see a persistent
block with AS 1, 1239, 7018, 5511, 2914, 3561, 6461, 3549, 3356, 701, 209, and 6453. These seem to
be mainly in US.
%(cf. \url{http://as-rank.caida.org/?mode0=as-dump-info}). 
(2) In the lower-right there is another smaller block with AS 8928, 286, 6695 and 13237,
which seem to be EU and DE.
%(3) It seem that the AS's in these two blocks are important ones, which also
%appear in here:
%\url{http://www.caida.org/research/topology/as_core_network/pics/ascore-2011-apr-ipv4v6-poster.pdf}
%(it's 2011 data but I can't find others.). (3)
(3) Between 2004 and 2005 there is some significant formation of new
communities. A similar phenomenon has been observed in ~\cite{edwards2012internet}.
%\url{http://www.cs.unm.edu/~csgsa/papers/2012/BenjaminEdwards.pdf}.
Moreover, by looking at the overlap structure, we find that that 
all the nodes in the US block mentioned above consistently appear in multiple clusters. These
turn out to be  Tier 1 providers or large internet exchange points.

%\underline{In Fig.~\ref{fig:skitter2}}: (1) For each snapshot, the number of nodes that are in multiple clusters are  
%$$ 7, 5, 4, 2 4, 1, 1, 19, 24, 0, 0, 13, 11, 8, 0, 0, 10, 10, 0, 5, 1, 4, 2, 4,
%5, 6, 7, 7, 7, 3, 11, 6$$
%I do not see any trend.

%Seems to make sense if we compare with \url{http://www.caida.org/research/topology/as_core_network/pics/ascore-2011-apr-ipv4v6-poster.pdf}.

\begin{figure}[!t]
\centering
\includegraphics[height=0.64\linewidth]{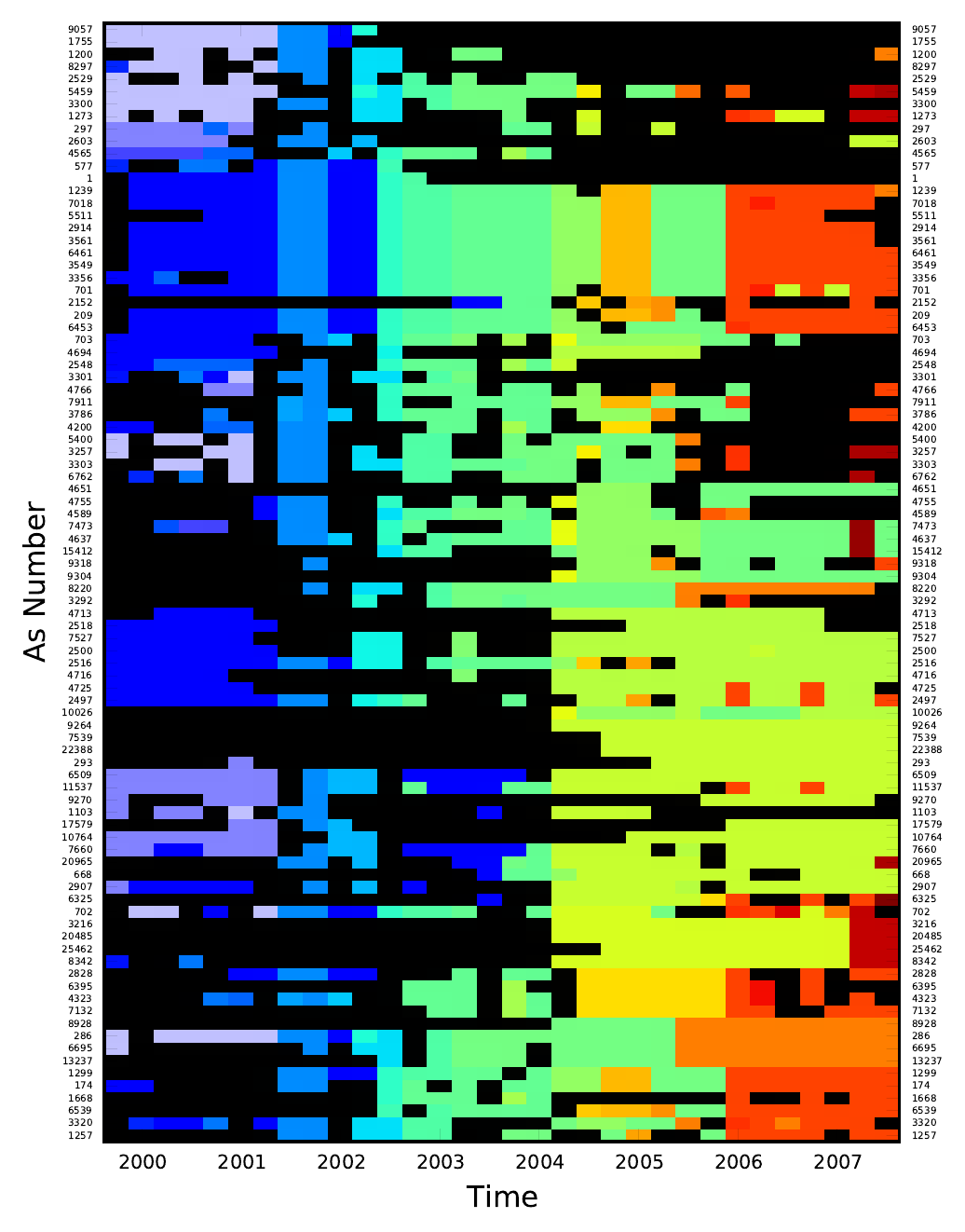}
\caption{Clustering result for the AS Link Dataset, overlaps not shown.}
\label{fig:skitter1}
\vspace{-0.15in}
\end{figure}

\iftechreport
\begin{figure}[!t]
\centering
\includegraphics[width = 1\linewidth]{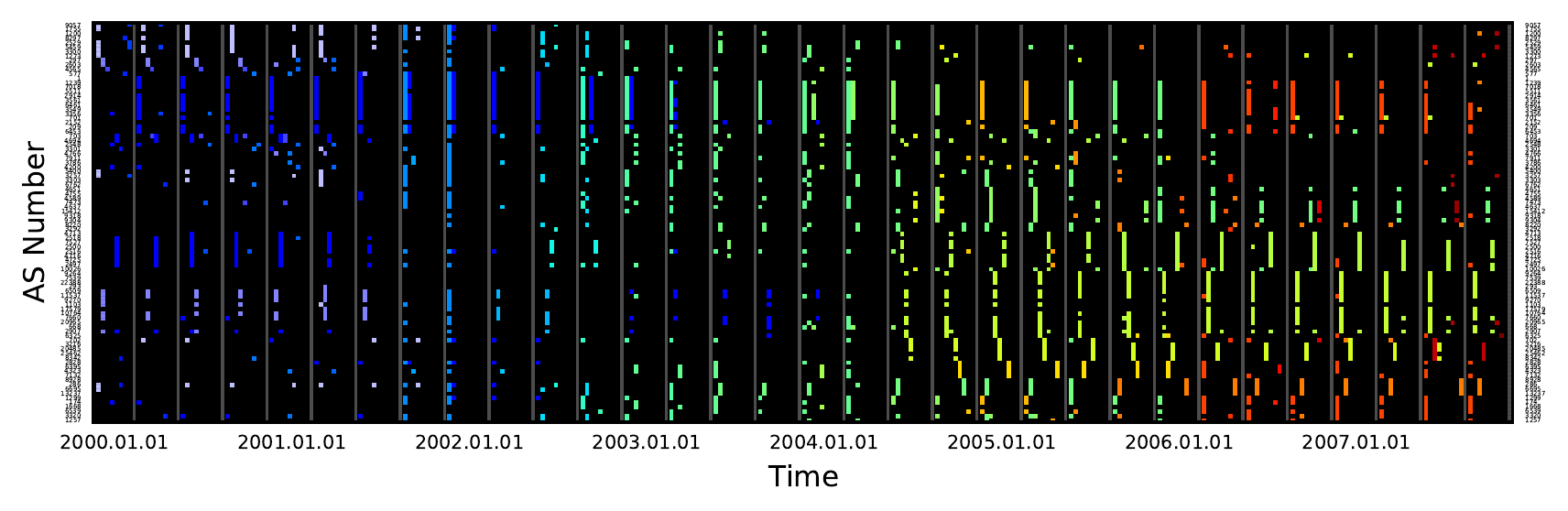}
\caption{Clustering with overlaps for the AS Link Dataset.}
\label{fig:skitter2}
\end{figure}
\fi

%% file: applications.tex
\section{Applications to Communication Networks}
\label{sec:apps}

We now describe some applications of
community detection to the design of communication networks.

\begin{comment}
We also discuss how the community detection method presented in 
this paper goes beyond the techniques that have been proposed in the literature 
for these applications and the resulting benefits.

\end{comment}

%%==========================================

%\subsection{DTN Routing}

\begin{comment}
Delay/Disruption Tolerant Networks (DTNs) are characterized by a lack of persistent 
network connectivity due to rapid changes in the network topology over time. As a 
result, traditional approaches for routing that require building and maintaining 
end-to-end paths are ill suited for such networks. Most routing schemes that have been 
proposed for DTNs rely on the ``store, carry, and forward'' paradigm that takes 
advantage of chance encounters between nodes to opportunistically forward packets 
towards their destinations. The efficacy of such forwarding mechanisms depends 
critically on their ability to select those nodes as relays that are more likely to 
encounter the destination.
\end{comment}

\textbf{Routing in Disruption Tolerant Networks (DTNs):} DTNs are often formed by devices that are carried around by humans whose mobility 
patterns are strongly influenced by their social relationships. Thus, 
%it can be expected  that 
the structures of the social graph between the humans and the the contact graph between the devices are correlated.
While the contact graph can change rapidly, it usually possesses a relatively stable underlying structure that is 
a function of the less volatile social graph~\cite{Daly2007, bubble, Gao2009, Mtibaa2010, Pietilanen2012}. 
This can be used to develop  ``social aware'' routing strategies that use social metrics such as node 
centrality and community labels to make forwarding decisions~\cite{DTN_survey, Daly2007, bubble, Gao2009, Mtibaa2010 }. All of these schemes utilize some form of community detection on the contact graph to 
infer social relations between the mobile nodes. However, the community detection 
methods used are generally limited to the non-overlapping and even non time-varying case.  The 
community detection framework proposed in this paper can be used with any of these 
schemes while overcoming these limitations. This can result in significant performance 
gains when, for example, people  belong to multiple social groups (e.g., 
friends, family, co-workers, etc.).

\textbf{Efficient Caching in Content Centric Networks:} Content based networking is an emerging paradigm that does not require connection 
oriented protocols between producers and consumers of information in communication 
networks. 
%Rather, it makes use of a publish-subscribe mechanism based on in-network 
%storage of named content. 
%Such a mechanism is particularly useful in applications such as 
%tactical MANETs that are characterized by rapid changes in topology and predominantly 
%multicast traffic. In particular, the same piece of content may be requested by multiple 
%users over time. 
Intelligent caching and replication of the content can 
significantly reduce access delays as well as the overhead costs associated with 
repeated querying and duplicate transmissions.
Recent work~\cite{slinky} proposes making use of the community structure of a 
MANET to determine nodes for content replication.  Assuming that the community 
structure changes on a slower time scale than the network topology, nodes in the
same community can cooperate to
provide an efficient and speedy access to content. 
%\cite{slinky} used heuristics to find the community
%structure and its evolution.
The method proposed in this paper can
provide a principled approach to build distributed content caching protocols.
%We also note that, unlike \cite{slinky}, our method allows overlapped nodes that
%are that are excellent candidates for caching as they can provide the content across a 
%larger number of users.
%and thus can be used to improve the performance of the caching scheme  in \cite{slinky}.

%However, the  strategy in \cite{slinky} does not allow overlapped community structure. Overlapped 
%nodes are excellent candidates for caching as they can provide the content across a 
%larger number of users. The community detection method proposed in this paper allows 
%overlapped nodes and thus can be used to improve the performance of the caching scheme 
%in \cite{slinky}.

%Note: The community used in \cite{slinky} refers to the contact graph community rather 
%than social graph one.

%Tactical Mobile Ad Hoc Networks (MANETs) share many common characteristics with DTNs 
%such as rapid changes in topology and lack of persistent end-to-end connectivity.

%%==========================================

%\subsection{Developing Realistic Mobility Models}
\textbf{Developing Realistic Mobility Models:} Much initial work on the design and analysis of routing algorithms for
mobile networks assumed simplistic mobility models such as random walk, random 
waypoint, etc. However, the analysis of mobility traces from many real-life scenarios
suggests that these simplistic models do not capture the details of real-world mobility 
characteristics such as periodicity and correlations due to social relationships between nodes.
Recent work on mobility modeling \cite{Hsu2009, Hossmann2011}
attempts to capture the dependence of the social relationships between nodes on their mobility
patterns. Community detection methods such as ours can be used to construct more refined 
mobility models that capture complex features such as the existence of overlapped 
communities as well as small yet persistent temporal communities.

%% file: conclusion.tex
\section{Conclusion}

In this paper, we consider the problem of detecting overlapping temporal
communities in dynamic networks.  A convex optimization based approach is
proposed for this problem. Theoretical and experimental results show that our
method is capable of revealing interesting community structure that cannot be
detected by methods that do not allow overlap, or those that do not utilize
temporal information.  For simplicity, in this work we have focused on
unweighted graphs. In the future, we plan to extend our method to treat weighted
graphs as well as develop distributed versions of the algorithm. We believe our
methods have wide applications in studying the structure and evolution of
complex networked systems including communication networks and social networks.

%% file: fast_algorithm.tex
\subsection{Fast Algorithm}\label{sec:fast_algo}

Solving the program \eqref{eq:cvx_opt} using standard package is feasible only
for small or medium size problems. In this section, we describe
a faster algorithm that is suitable for larger scale datasets. Our
method is based on matrix factorization. Each positive semidefinite
cover matrix $Y^{t}$ can be factorized as $Y^{t}=U^{t}U^{t\top}$, where
$U^{t}\in\mathbb{R}^{n\times r}$ and $\left\Vert Y^{t}\right\Vert _{*}=\left\Vert U^{t}\right\Vert _{F}^{2}$. Here $ r $ is any upper-bound on the number of clusters at each snapshot; one can always use $ r=n $, but a smaller value is more desirable.
We consider the Lagrangian of the original constrained formulation.
The optimization problem becomes 
\begin{align}
\max_{U^{t}} \;
\sum_{t=1}^{T} \! f(U^{t}U^{t\top}|A^{t}) \!-\! \gamma\! \sum_{t=1}^{T-1}\! d(Y^{t+1},Y^{t}),\;\;
\textrm{s.t.} \left\Vert U^{t}\right\Vert _{F}^{2} \!\le\! B.\label{eq:fac_opt}
\end{align}
Choosing the multiplier $\gamma $ is equivalent to choosing $ \delta $ in the original formulation. We use (sub-)gradient descent to solve the problem:
\begin{align}
U^{t}\leftarrow\mathcal{P}_{\sqrt{B}}&\left[ U^{t}+\tau^t\left(\nabla f(U^{t}U^{t\top})-\gamma\nabla_{2}d(U^{t+1}U^{t+1\top},U^{t}U^{t\top})\right.\right.\nonumber\\
&\;\left.\left.-\gamma\nabla_{1}d(U^{t}U^{t\top},U^{t-1}U^{t-1})\right)U^{t}\right],\label{eq:gradient}
\end{align}
where $ \nabla $ is the sub-gradient operator,  $ \nabla_i $ denotes the (sub-)gradient w.r.t. the $ i $-th argument,  $\mathcal{P}_{\sqrt{B}}(X)$ is the Euclidean projection of $X$ onto
the Frobenius norm ball $\left\{ Z:\left\Vert Z\right\Vert _{F}\le\sqrt{B}\right\} $
(i.e., scale down $X$ to have Frobenius norm $\sqrt{B}$ if and only if it is
outside the ball), and $ \tau^t $ is the step size.  As for all gradient descent methods, the above procedure is guaranteed to converge provided $ \tau^t \rightarrow 0 $. In this paper, we use a geometrically decreasing step size $ \tau^t = 0.001\cdot 0.995^t $. 

\subsection{Complexity and Scalability}\label{sec:complexity}

We analyze the memory and time complexity of the gradient descent algorithm.

\textbf{Memory complexity}: We need to store the adjacency matrices $\{A^t\}$
and the factorization $\{U^t\}$, which requires
$O(E)$ and $O(nrT)$ memory, respective; here $E$ is the total number of edges in all snapshots,
$n$ the number of nodes, $r$ the maximum number of clusters at each snapshot, 
and $T$ number of snapshots. The total memory complexity is $O(E+nrT)$.
The online implementation suggested in Section \ref{sec:algorithm_remark} will further alleviate the
dependence on $T$.

\textbf{Time complexity}: The algorithm requires time $T_1$ for
computing a initial point, and  $T_2$ for
each iteration with $M$ iterations.
Here we initialize $U^t$ by taking a rank-$r$ SVD of $A^t$. For each $t$ this can
be done in time $O(rE_t + nr^2)$ (see
~\cite{halko2011finding}), where $E_t$ is the number of edges in snapshot $t$.
So $T_1=O(rE + nr^2T)$. Now consider the update \eqref{eq:gradient}.
The computation of the product of three (sub-)gradient operators with $U^t$ takes time
$O(r^2 E_t+ nr)$, $O(r^2E_t)$ and $O(r^2E_t)$, respectively, by taking advantage of 
the fact that we can use any sub-gradient. The summation and
the projection both take $O(nr)$. We thus have $T_2= O(r^2E+nrT)$.  The total time
complexity is then $O(nr^2T+ M r^2E+ MnrT)$. Characterizing the number of iterations $M$ 
needed for a specified accuracy rigorously is difficult, However, as observed empirically in our simulations and many other
studies, $ M $ is independent of $ E $, $n$ and $T$, and can be treated as $ O(1) $.

In summary, with a bounded number of clusters $r$, both the space
and time complexity scale linearly in $E$ and $nT$. This is the best one can hope
for, as it takes at least this much space and time to read the input and write
down the final solution.

%\subsection{Relation to Modularity}
%
%In this section we show that the quality function we use is closely related to the modularity function when there is no overlap and the graph is unweighted.
%
%Recall that the modularity is the sum of $ n^2 $ terms, each term being $ \left( \frac{k_i k_j  }{2M}- A_{ij}\right) Y_{ij} $. If $\frac{k_i k_j}{2M} < A_{ij}$, then the coefficient before $ Y_{ij} $ is positive and it encourages $ i $ and $ j $ to be assigned to the same cluster ($ Y_{ij} =1$); it incurs a cost of $ \left| \frac{k_i k_j  }{2M}- A_{ij}\right|$ if $ Y_{ij}=0 $. Similarly, if $\frac{k_i k_j}{2M} > A_{ij}$, then it favors $ Y_{ij} =0$ and penalizes $ Y_{ij}=0 $ again by $ \left| \frac{k_i k_j  }{2M}- A_{ij}\right|$. Therefore, modularity tries to put two nodes in the same clusters \emph{only when} $ A_{ij} $ exceeds the threshold $\frac{k_i k_j}{2M}$ (which is the expected number of edges between $ i $ and $ j $ in the null model). If we define the thresholded version of $ A $ as $ \bar{A}_{ij} = \mathbb{I}(A_{ij}>\frac{k_i k_j}{2M}) $, where $ \mathbb{I}(\cdot) $ is the indicator function, then it is not hard to see that modularity maximization is equivalent to maximizing $ - \sum_{i,j}\left| \frac{k_i k_j  }{2M}- A_{ij}\right| |Y_{ij} - \bar{A}_{ij}|$.
%
%For unweighted networks, we usually have $ \bar{A}_{ij} = A_{ij}$ for almost all pairs of nodes. This is true, e.g., in the standard planted partition model \cite{condon2001algorithms}. When this happens, modularity is identical to our quality function.

%% file: proof.tex
\subsection{Proof of Theorem \ref{thm:no_jump}}\label{sec:proof_of_theorem}

The following lemma shows that it suffices to study the Lagrangian formulation. Recall that $ \Vert X \Vert_1 =\sum_{i,j} |X_{ij}|$ is the matrix $ \ell_1 $ norm of $ M $. Let $ \circ $ denote the entry-wise product.
\begin{lemma}
$ Y^* $ is the unique optimal solution to \eqref{eq:no_jump_opt} if there exists a $ \lambda $ such that $ Y^* $ is the unique optimal solution to the following problem
\begin{equation}
\min_{Y}\left\Vert Y\right\Vert _{*}+\lambda\sum_{t} \Vert C\circ (Y-A^t)\Vert_1\label{eq:no_jump_opt_lang}
\end{equation}
\end{lemma}
\begin{proof}
Let $ g(Y) =\left\Vert Y\right\Vert _{*}  $ and $ h(Y)= \sum_{t} \Vert C\circ (Y-A^t)\Vert_1$. Note that $g(Y^*)=n$. By standard convex analysis and the fact that $ Y^* $ is optimal to \eqref{eq:no_jump_opt_lang}, we have the following chain of inequality:
\begin{align*}
\eqref{eq:no_jump_opt} & = \min_{Y:g(Y)\le n} h(Y) = \max_{\lambda'} \min_{Y} h(Y) + \frac{1}{\lambda'}(g(Y)-n)  \\
& \ge \frac{1}{\lambda} \min_Y  \lambda h(Y) + (g(Y)-n) = h(Y^*)\ge \eqref{eq:no_jump_opt}.
\end{align*}
Therefore, equality holds above, which proves that $ Y^* $ is \emph{an} optimal
solution to \eqref{eq:no_jump_opt}. We prove uniqueness by contradiction. If $ Y^* $ is not the unique optimal solution to \eqref{eq:no_jump_opt}, then there exists $ Y' $ with $ g(Y') \le n $ and $ h(Y') = \eqref{eq:no_jump_opt} $. Using the equality we just proved, we have 
$$ 
h(Y') + \lambda(g(Y')-n )  \le h(Y') = \eqref{eq:no_jump_opt} = \frac{1}{\lambda} \min_Y  \lambda h(Y) + (g(Y)-n),
$$
which contradicts the assumption that $ Y^* $ is the unique optimal solution to \eqref{eq:no_jump_opt_lang}.
\end{proof}

To prove Theorem \ref{thm:no_jump}, it  suffices to show that $ Y^* $ is the unique optimal solution to \eqref{eq:no_jump_opt_lang} with $\lambda=\sqrt{\frac{1}{16mn}}$. We do this by showing that any other solution $Y^{*}+\Delta$ with $\Delta\neq0$
has a higher objective value. 

We define a matrix $W$ which serves
as a dual certificate. Let $S^{t}=A^{t}-Y^{*}$, $\Omega_{t}=\left\{ (i,j)|S^{t}_{ij}\neq0\right\} $, $ R=\{(i,j)|Y_{ij}=1\} $,
and $U$ be the matrix whose columns are the singular vectors of $Y^{*}$.
For any entry set $ \Omega \subseteq [n]\times[n]$, let $ 1_\Omega $ denote the matrix whose entries in $ \Omega $ equals $ 1 $ and others $ 0 $. 
Define $W=\sum_{t=1}^{m}V^{t}+\sum_{t}Z^{t}$, where 
\begin{align*}
V^{t} & =  \frac{1}{m}\left(-P_{\Omega_{t}}UU^{\top}+\frac{1-p}{p}P_{\Omega_t^{c}}UU^{\top}\right)\\
Z^{t} & =  2\lambda\left(C\circ S^t +  \frac{1-p}{p}\sum_{(i,j) \in R\cap\Omega_{t}^{c}} (1-s)1_{(i,j)}  \right.\\
& \left.\qquad\;\quad - \frac{q}{1-q} \sum_{(i,j) \in R^c\cap\Omega_{t}^{c}} s 1_{(i,j)} \right).
\end{align*}
Due to the randomness in $ \Omega^t $, both $\sum_{t}V^{t}$ and $\sum_{t}Z^{t}$ are random matrices with
independent zero-mean entries, whose variances are bounded by $\frac{1}{K^{2}m}$
and $4\lambda^{2}m$ due to the setup of the model. Under our choice of $ \lambda $ and the assumption of the theorem,
they are further bounded by $\frac{1}{2n}$. Let $ \Vert \cdot \Vert $ be the spectral norm (the largest singular value). Standard bounds
on the spectral norm of random matrices guarantees that with probability converging to one,
\[
\left\Vert P_{T^{\bot}}W\right\Vert \le\left\Vert \sum_{t}V^{t}\right\Vert +\left\Vert \sum_{t}Z^{t}\right\Vert \le1.
\]
It follows that $UU^{\top}+P_{T^{\bot}}W$ is a subgradient of $\left\Vert Y\right\Vert _{*}$,
which means $\left\langle Y^{*}+\Delta,UU^{\top}+P_{T^{\bot}}W\right\rangle \ge\left\langle \Delta,UU^{\top}+P_{T^{\bot}}W\right\rangle $
for all $\Delta$. Also define $F^{t}=-\mbox{sign}(P_{\Omega_t^{c}}(\Delta^{t}))$,
where $\mbox{sign}(\cdot)$ is the signum function, so $\left\langle F^{t},\Delta^{t}\right\rangle =\left\Vert P_{\Omega_{t}}\Delta^{t}\right\Vert _{1}$.
We also know $C\circ (S^{t}+F^{t})$ is a subgradient of $\left\Vert C \circ S^{t}\right\Vert _{1}$,
so $\left\Vert C\circ (S^{t}-\Delta)\right\Vert _{1}-\left\Vert C\circ S^{t}\right\Vert _{1}\ge\left\langle C\circ (S^t+F^t),-\Delta\right\rangle $.
Combining the above discussion, we have
\begin{align*}
% &   \left\Vert Y^{*}+\Delta\right\Vert _{*}-\left\Vert Y\right\Vert _{*}+\lambda\sum_{t}\left(\left\Vert C\circ (A^{t}-Y^{*}-\Delta)\right\Vert _{1}-\left\Vert C \circ (A^{t}-Y^{*})\right\Vert _{1}\right)\\
 & \left\Vert Y+\Delta\right\Vert _{*}-\left\Vert Y\right\Vert_* +\lambda\sum_{t}\left(\left\Vert C\circ (S^{t}-\Delta)\right\Vert _{1}-\left\Vert C\circ S^{t}\right\Vert _{1}\right)\\
 \ge &  \left\langle UU^{\top}+P_{T^{\bot}}W,\Delta\right\rangle +\lambda\sum_{i}\left\langle C\circ (S^t+F^t),-\Delta\right\rangle 
\end{align*}
We bound each of the above two terms. Notice that 
\begin{align*}
 & \left\langle UU^{\top}+P_{T^{\bot}}W,\Delta\right\rangle  
 =   \left\langle UU^{\top}+W,\Delta\right\rangle -\left\langle P_{T}W,\Delta\right\rangle \\
 = & \sum_{t}\left\langle (P_{\Omega_{t}}+P_{\Omega^{c}_t})(\frac{1}{m}UU^{\top}+V^{t}+Z^{t}),\Delta\right\rangle 
  -\left\langle P_{T}W,\Delta\right\rangle \\
 \ge &  2\lambda\sum_{t}\left\Vert P_{\Omega_{t}}(C\circ \Delta)\right\Vert _{1} 
 -\left\Vert P_{T}W\right\Vert _{\infty}\left\Vert \Delta\right\Vert _{1}\\
 &  +\sum_{t}\left\langle \frac{1}{m}\frac{1}{1-q}P_{\Omega_{t}^{c}}UU^{\top}+ 2\lambda\frac{1-p}{p}\sum_{(i,j) \in R\cap\Omega_{t}^{c}} (1-s)1_{(i,j)}  \right.\\
 & \quad\quad \left.- 2\lambda\frac{q}{1-q} \sum_{(i,j) \in R^c\cap\Omega_{t}^{c}} s 1_{(i,j)},\Delta\right\rangle;
\end{align*}
here $ \Vert M\Vert_\infty := \max_{i,j} |M_{ij}| $ is the matrix $ \ell_\infty $ norm. Under the assumption of the theorem, we have 
\begin{align*}
 & \left\langle \frac{1}{m}\frac{1}{1-q}P_{\Omega_{t}^{c}}UU^{\top} + 2\lambda\frac{1-p}{p}\sum_{(i,j) \in R\cap\Omega_{t}^{c}} (1-s)1_{(i,j)},\Delta\right\rangle \\
 \ge & -\frac{1}{2}\lambda\left\Vert P_{R\cap\Omega_t^{c}} (C\circ \Delta)\right\Vert _{1}
\end{align*}
and 
\begin{align*}
\left\langle- 2\lambda\frac{q}{1-q} \sum_{(i,j) \in R^c\cap\Omega_{t}^{c}} s 1_{(i,j)},\Delta\right\rangle \ge-\frac{1}{2}\lambda\left\Vert P_{R^{c}\cap\Omega^{tc}} (C\circ\Delta) \right\Vert _{1}.
\end{align*}
Moreover, observe that each entry of $\left(P_{T}W\right)_{ij}=\frac{1}{K}\sum_{k=1}^{K}W_{ij}$,
which is the sum of independent random variables with bounded variance as previously  discussed. Under the assumption of the theorem, this sum is bounded
by $\frac{1}{\sqrt{K}}\sqrt{\frac{1}{2n}}\le\frac{1}{4}m\lambda \min\{s,1-s\}$
with probability converging to one by standard Bernstein inequality; $ \Vert P_T W\Vert_{\infty} $ is bounded by the same quantity using a union bound. It follows
that
\[
\left\langle UU^{\top}+P_{T^{\bot}}W,\Delta\right\rangle \ge\frac{7}{4}\lambda\sum_{t}\left\Vert P_{\Omega_{t}}(C\circ \Delta)\right\Vert _{1}-\frac{3}{4}\lambda\sum_{t}\left\Vert P_{\Omega_t^{c}}(\circ \Delta)\right\Vert _{1}
\]
On the other hand, we have 
\begin{align*}
 & \lambda\sum_{t}\left\langle C\circ (S_{t}+F_{t}),-\Delta\right\rangle \\
 = & -\lambda\sum_{t}\left\Vert P_{\Omega_{t}}(C\circ \Delta)\right\Vert _{1}+\lambda\sum_{t}\left\Vert P_{\Omega_{t}^{c}}(C\circ \Delta)\right\Vert _{1}.
\end{align*}
Combining pieces, we obtain
\begin{align*}
   & \left\Vert Y^{*}+\Delta\right\Vert _{*}-\left\Vert Y\right\Vert _{*}+\lambda\sum_{t}\left(\left\Vert C\circ (S^t-\Delta)\right\Vert _{1}-\left\Vert C\circ S^t\right\Vert _{1}\right)\rangle \\
  \ge & \left\langle UU^{\top}+P_{T^{\bot}}W,\Delta\right\rangle - \lambda\sum_{t}\left\Vert P_{\Omega^{t}}(C\circ \Delta)\right\Vert _{1}+\lambda\sum_{t}\left\Vert P_{\Omega^{tc}}(C\circ \Delta)\right\Vert _{1}\\
  \ge & \frac{7}{4}\lambda\sum_{t}\left\Vert P_{\Omega^{t}}(C\circ \Delta)\right\Vert_{1}-\frac{3}{4}\lambda\sum_{t}\left\Vert P_{\Omega^{tc}}(C\circ \Delta)\right\Vert_{1} \\
  & -\lambda\sum_{t}\left\Vert P_{\Omega^{t}}(C\circ \Delta)\right\Vert_{1}+\lambda\sum_{t}\left\Vert P_{\Omega^{tc}}(C\circ \Delta)\right\Vert_{1}\\
  > & 0.
\end{align*}
This completes the proof of the theorem.

%% file: arxiv_temporal_clustering.bbl
\begin{thebibliography}{10}
\providecommand{\url}[1]{#1}
\csname url@samestyle\endcsname
\providecommand{\newblock}{\relax}
\providecommand{\bibinfo}[2]{#2}
\providecommand{\BIBentrySTDinterwordspacing}{\spaceskip=0pt\relax}
\providecommand{\BIBentryALTinterwordstretchfactor}{4}
\providecommand{\BIBentryALTinterwordspacing}{\spaceskip=\fontdimen2\font plus
\BIBentryALTinterwordstretchfactor\fontdimen3\font minus
  \fontdimen4\font\relax}
\providecommand{\BIBforeignlanguage}[2]{{%
\expandafter\ifx\csname l@#1\endcsname\relax
\typeout{** WARNING: IEEEtran.bst: No hyphenation pattern has been}%
\typeout{** loaded for the language `#1'. Using the pattern for}%
\typeout{** the default language instead.}%
\else
\language=\csname l@#1\endcsname
\fi
#2}}
\providecommand{\BIBdecl}{\relax}
\BIBdecl

\bibitem{Fortunato2010community}
S.~Fortunato, ``Community detection in graphs,'' \emph{Physics Reports}, vol.
  486, no. 3-5, pp. 75--174, 2010.

\bibitem{evolutionaryclustering}
D.~Chakrabarti, R.~Kumar, and A.~Tomkins, ``Evolutionary clustering,'' in
  \emph{ACM KDD}, 2006.

\bibitem{Facetnet}
Y.-R. Lin, Y.~Chi, S.~Zhu, H.~Sundaram, and B.~L. Tseng, ``Facetnet: a
  framework for analyzing communities and their evolutions in dynamic
  networks,'' in \emph{ACM WWW}, 2008.

\bibitem{kim2009particle}
M.~Kim and J.~Han, ``A particle-and-density based evolutionary clustering
  method for dynamic networks,'' \emph{Proceedings of the VLDB Endowment},
  vol.~2, no.~1, pp. 622--633, 2009.

\bibitem{estrangement}
V.~Kawadia and S.~Sreenivasan, ``Sequential detection of temporal communities
  by estrangement confinement,'' \emph{Scientific Reports 2}, 2012.

\bibitem{slinky}
V.~Kawadia, N.~Riga, J.~Opper, and D.~Sampath, ``Slinky: An adaptive protocol
  for content access in disruption-tolerant ad hoc networks,'' in \emph{ACM
  Intl. Workshop on Tactical Mobile Ad Hoc Networking}, 2011.

\bibitem{chen2012overlap}
Y.~Chen, H.~Xu, and S.~Sanghavi, ``Graph clustering with overlaps,''
  \emph{Manuscript. Submitted}.

\bibitem{agarwal2008modularity}
G.~Agarwal and D.~Kempe, ``Modularity-maximizing graph communities via
  mathematical programming,'' \emph{The European Physical Journal B}, vol.~66,
  no.~3, pp. 409--418, 2008.

\bibitem{Newman2006}
M.~E.~J. Newman, ``Modularity and community structure in networks,''
  \emph{PNAS}, vol. 103, no.~23, pp. 8577--8582, 2006.

\bibitem{candes2009robustPCA}
E.~Candes, X.~Li, Y.~Ma, and J.~Wright, ``{Robust principal component
  analysis?}'' \emph{Preprint arXiv:0912.3599}, 2009.

\bibitem{chandrasekaran2011siam}
V.~Chandrasekaran, S.~Sanghavi, S.~Parrilo, and A.~Willsky, ``Rank-sparsity
  incoherence for matrix decomposition,'' \emph{SIAM Journal on Optimization},
  vol.~21, no.~2, pp. 572--596, 2011.

\bibitem{chen2012sparseclustering}
Y.~Chen, S.~Sanghavi, and H.~Xu, ``Clustering sparse graphs,'' \emph{Advances
  in neural information processing systems 25}, 2012.

\bibitem{overlapping2012Survey}
J.~Xie, S.~Kelley, and B.~K. Szymanski, ``Overlapping community detection in
  networks: the state of the art and comparative study,'' \emph{To appear in
  ACM Computing Surveys}, vol. abs/1110.5813, 2011.

\bibitem{cazabet2010overlap}
R.~Cazabet, F.~Amblard, and C.~Hanachi, ``Detection of overlapping communities
  in dynamical social networks,'' in \emph{IEEE SocialCom}, 2010.

\bibitem{nguyen2011adaptive}
N.~Nguyen, T.~Dinh, Y.~Xuan, and M.~Thai, ``Adaptive algorithms for detecting
  community structure in dynamic social networks,'' in \emph{INFOCOM, 2011
  Proceedings IEEE}.\hskip 1em plus 0.5em minus 0.4em\relax IEEE, 2011, pp.
  2282--2290.

\bibitem{nguyen2011overlapping}
N.~Nguyen, T.~Dinh, S.~Tokala, and M.~Thai, ``Overlapping communities in
  dynamic networks: their detection and mobile applications,'' in \emph{ACM
  Mobicom}, 2011, pp. 85--96.

\bibitem{Graphscope}
J.~Sun, C.~Faloutsos, S.~Papadimitriou, and P.~S. Yu, ``Graphscope:
  parameter-free mining of large time-evolving graphs,'' in \emph{ACM KDD},
  2007.

\bibitem{timefall}
J.~Ferlez, C.~Faloutsos, J.~Leskovec, D.~Mladenic, and M.~Grobelnik,
  ``Monitoring network evolution using {MDL},'' in \emph{IEEE ICDE}, 2008.

\bibitem{Mucha2010}
P.~Mucha, T.~Richardson, K.~Macon, M.~Porter, and J.~Onnela, ``{Community
  structure in time-dependent, multiscale, and multiplex networks},''
  \emph{Science}, vol. 328, no. 5980, pp. 876--878, 2010.

\bibitem{chi2009evolutionary}
Y.~Chi, X.~Song, D.~Zhou, K.~Hino, and B.~Tseng, ``On evolutionary spectral
  clustering,'' \emph{ACM Trans on Knowledge Discovery from Data}, vol.~3,
  no.~4, p.~17, 2009.

\bibitem{bansal2004correlation}
N.~Bansal, A.~Blum, and S.~Chawla, ``Correlation clustering,'' \emph{Machine
  Learning}, vol.~56, no.~1, pp. 89--113, 2004.

\bibitem{ames2011clique}
B.~Ames and S.~Vavasis, ``Nuclear norm minimization for the planted clique and
  biclique problems,'' \emph{Mathematical Programming}, vol. 129, no.~1, pp.
  69--89, 2011.

\bibitem{recht2010guaranteed}
B.~Recht, M.~Fazel, and P.~Parrilo, ``{Guaranteed Minimum-Rank Solutions of
  Linear Matrix Equations via Nuclear Norm Minimization},'' \emph{SIAM Review},
  vol.~52, no. 471, 2010.

\bibitem{seung2001nmf}
D.~Seung and L.~Lee, ``Algorithms for non-negative matrix factorization,''
  \emph{Advances in neural information processing systems 13}, 2000.

\bibitem{condon2001algorithms}
A.~Condon and R.~Karp, ``Algorithms for graph partitioning on the planted
  partition model,'' \emph{Random Structures and Algorithms}, vol.~18, no.~2,
  pp. 116--140, 2001.

\bibitem{lakehursttrace}
Y.~Lu, F.~Wicker, Y.~Chen, P.~Li{\'o}, and D.~Towsley, ``On secure network
  structures in the lakehurst trace,'' 2008.

\bibitem{Eagle}
N.~Eagle and A.~Pentland, ``Reality mining: sensing complex social systems,''
  \emph{Personal Ubiquitous Comput.}, vol.~10, pp. 255--268, 2006.

\bibitem{COWTrade}
\BIBentryALTinterwordspacing
K.~Barbieri and O.~Keshk, ``Correlates of war project trade data set.''
  [Online]. Available: \url{correlatesofwar.org}
\BIBentrySTDinterwordspacing

\bibitem{skitter}
B.~Huffaker, Y.~Hyun, D.~Andersen, and kc~claffy, ``The {Skitter AS Links
  Dataset}.''

\bibitem{edwards2012internet}
B.~Edwards, S.~Hofmeyr, G.~Stelle, and S.~Forrest, ``Internet topology over
  time,'' \emph{Preprint arXiv:1202.3993}, 2012.

\bibitem{Daly2007}
E.~M. Daly and M.~Haahr, ``Social network analysis for routing in disconnected
  delay-tolerant manets,'' in \emph{ACM Mobihoc}, 2007, pp. 32--40.

\bibitem{bubble}
P.~Hui, J.~Crowcroft, and E.~Yoneki, ``Bubble rap: Social-based forwarding in
  delay-tolerant networks,'' \emph{Mobile Computing, IEEE Transactions on},
  vol.~10, no.~11, pp. 1576 --1589, nov. 2011.

\bibitem{Gao2009}
W.~Gao, Q.~Li, B.~Zhao, and G.~Cao, ``Multicasting in delay tolerant networks:
  a social network perspective,'' in \emph{ACM MobiHoc}, 2009.

\bibitem{Mtibaa2010}
A.~Mtibaa, M.~May, C.~Diot, and M.~Ammar, ``Peoplerank: social opportunistic
  forwarding,'' in \emph{IEEE INFOCOM}, 2010.

\bibitem{Pietilanen2012}
A.-K. Pietil\"{a}nen and C.~Diot, ``Dissemination in opportunistic social
  networks: the role of temporal communities,'' in \emph{ACM MobiHoc}, 2012.

\bibitem{DTN_survey}
Y.~Zhu, B.~Xu, X.~Shi, and Y.~Wang, ``A survey of social-based routing in delay
  tolerant networks: Positive and negative social effects,'' \emph{IEEE Comm.
  Surveys Tutorials}, vol.~PP, no.~99, pp. 1--15, 2012.

\bibitem{Hsu2009}
W.-J. Hsu, T.~Spyropoulos, K.~Psounis, and A.~Helmy, ``Modeling spatial and
  temporal dependencies of user mobility in wireless mobile networks,''
  \emph{IEEE/ACM Trans. Netw.}, vol.~17, pp. 1564--1577, 2009.

\bibitem{Hossmann2011}
T.~Hossmann, T.~Spyropoulos, and F.~Legendre, ``Putting contacts into context:
  mobility modeling beyond inter-contact times,'' in \emph{ACM MobiHoc}, 2011,
  pp. 18:1--18:11.

\bibitem{halko2011finding}
N.~Halko, P.~Martinsson, and J.~Tropp, ``Finding structure with randomness:
  Probabilistic algorithms for constructing approximate matrix
  decompositions,'' \emph{SIAM review}, vol.~53, no.~2, pp. 217--288, 2011.

\end{thebibliography}
